\documentclass[pra,superscriptaddress,twocolumn,showpacs,10pt]{revtex4}

\usepackage{amsmath}
\usepackage{amsthm}
\usepackage{amssymb}
\usepackage{latexsym}
\usepackage{float}
\usepackage{graphicx}
\usepackage{enumerate}
\usepackage{enumitem}
\usepackage{mathtools}
\usepackage{afterpage}

\usepackage{xy}
\xyoption{matrix}
\xyoption{frame}
\xyoption{arrow}
\xyoption{arc}

\usepackage{ifpdf}
\ifpdf
\else
\xyoption{ps}
\xyoption{dvips}
\fi

\entrymodifiers={!C\entrybox}

\newcommand{\bra}[1]{{\langle{#1}\vert}}
\newcommand{\ket}[1]{{\vert{#1}\rangle}}
\newcommand{\qw}[1][-1]{\ar @{-} [0,#1]}
\newcommand{\qwx}[1][-1]{\ar @{-} [#1,0]}

\newcommand{\gate}[1]{*+<.6em>{#1} \POS ="i","i"+UR;"i"+UL **\dir{-};"i"+DL **\dir{-};"i"+DR **\dir{-};"i"+UR **\dir{-},"i" \qw}

\newcommand{\meter}{*=<1.8em,1.4em>{\xy ="j","j"-<.778em,.322em>;{"j"+<.778em,-.322em> \ellipse ur,_{}},"j"-<0em,.4em>;p+<.5em,.9em> **\dir{-},"j"+<2.2em,2.2em>*{},"j"-<2.2em,2.2em>*{} \endxy} \POS ="i","i"+UR;"i"+UL **\dir{-};"i"+DL **\dir{-};"i"+DR **\dir{-};"i"+UR **\dir{-},"i" \qw}

\newcommand{\control}{*!<0em,.025em>-=-<.2em>{\bullet}}

\newcommand{\ctrl}[1]{\control \qwx[#1] \qw}

\newcommand{\multigate}[2]{*+<1em,.9em>{\hphantom{#2}} \POS [0,0]="i",[0,0].[#1,0]="e",!C *{#2},"e"+UR;"e"+UL **\dir{-};"e"+DL **\dir{-};"e"+DR **\dir{-};"e"+UR **\dir{-},"i" \qw}

\newcommand{\ghost}[1]{*+<1em,.9em>{\hphantom{#1}} \qw}

\newcommand{\lstick}[1]{*!R!<.5em,0em>=<0em>{#1}}

\newcommand{\Qcircuit}{\xymatrix @*=<0em>}

\newtheorem{theorem}{Theorem}

\newtheorem{lemma}{Lemma}

\newtheorem{corollary}{Corollary}

\theoremstyle{definition}

\newcommand{\qedsymb}{\hfill{\rule{2mm}{2mm}}}
\renewenvironment{proof}[1][]{\begin{trivlist} % Note changed to renewenvironment by Ben because I loaded amsthm package.
\item[\hspace{\labelsep}{\bf\noindent Proof#1:\/}] }{\qedsymb\end{trivlist}}

\newcommand{\beq}{\begin{equation}}
\newcommand{\eeq}{\end{equation}}

\newcommand{\bea}{\begin{eqnarray}}
\newcommand{\eea}{\end{eqnarray}}

\newcommand{\ba}{\begin{array}}
\newcommand{\ea}{\end{array}}

\newcommand{\bc}{\begin{cases}}
\newcommand{\ec}{\end{cases}}

\newcommand{\beba}{\begin{equation}\begin{array}{lll}}
\newcommand{\eeea}{\end{array}\end{equation}}

\newcommand{\bpm}{\begin{pmatrix}}
\newcommand{\epm}{\end{pmatrix}}

\newcommand{\bit}{\begin{itemize}}
\newcommand{\eit}{\end{itemize}}

\newcommand{\ben}{\begin{enumerate}}
\newcommand{\een}{\end{enumerate}}

\newcommand{\bracket}[2]{\langle #1|#2\rangle}
\newcommand{\ketbra}[2]{|#1\rangle \langle #2 |}

\newcommand{\ip}[2]{\langle #1, #2\rangle}

\newcommand{\tr}{\mathrm{tr}}

\newcommand{\supp}{\mathrm{supp}}

\newcommand{\HD}{\mathcal{H}_D}
\newcommand{\Hd}{\mathcal{H}_d}

\newcommand{\Hdn}{\mathcal{H}_d^{\otimes n}}
\newcommand{\Htn}{\mathcal{H}_2^{\otimes n}}

\newcommand{\BHD}{\mathcal{B}(\HD)}
\newcommand{\BHd}{\mathcal{B}(\Hd)}
\newcommand{\BHdn}{\mathcal{B}(\Hdn)}

\newcommand{\BHtn}{\mathcal{B}(\Htn)}

\newcommand{\CD}{\mathbb{C}^D}
\newcommand{\Pn}{\mathrm{P}_n}

\newcommand{\I}{\mathcal{I}}
\newcommand{\N}{\mathbb{N}}
\newcommand{\C}{\mathbb{C}}

\renewcommand{\S}{\mathcal{S}}
\newcommand{\T}{\mathcal{T}}

\newcommand{\dint}{\displaystyle\int}
\renewcommand{\Re}{\mathrm{Re}}

\newcommand{\E}{\mathrm{E}}

\newcommand{\ME}{\ket{\Phi^+_D}}

\newcommand{\poly}{\mathrm{poly}}

\newcommand{\Z}{\mathbb{Z}}
\newcommand{\Zd}{\mathbb{Z}_d}

\newcommand{\Zdn}{\mathbb{Z}_d^n}
\newcommand{\Ztn}{\mathbb{Z}_2^n}

\newcommand{\vx}{\mathbf{x}}
\newcommand{\vy}{\mathbf{y}}
\newcommand{\vz}{\mathbf{z}}
\newcommand{\va}{\mathbf{a}}
\newcommand{\vb}{\mathbf{b}}
\newcommand{\vc}{\mathbf{c}}
\newcommand{\vd}{\mathbf{d}}
\newcommand{\vzero}{\mathbf{0}}

\newcommand{\vp}{\mathbf{p}}
\newcommand{\vq}{\mathbf{q}}

\newcommand{\vL}{\mathbf{L}}

\newcommand{\vla}{\mathbf{\lambda}}

\newcommand{\sch}{\mathrm{Sch}}

\newcommand{\hook}{\mathrm{hook}}

\newcommand{\V}{\mathcal{V}}
\newcommand{\W}{\mathcal{W}}

\floatstyle{ruled}
\newfloat{algorithm}{htbp}{loa}
\floatname{algorithm}{Algorithm}

\begin{document}

\title{Property testing of quantum measurements}

\author{Guoming Wang}
\email{gmwang@eecs.berkeley.edu}
\affiliation{Computer Science Division, University of California Berkeley, Berkeley, California 94720, USA}

\date{\today}

\begin{abstract}
In this paper, we study the following question: given a black box performing some unknown quantum measurement on a multi-qudit system, how do we test whether this measurement has certain property or is far away from having this property. We call this task \textit{property testing} of quantum measurement. We first introduce a metric for quantum measurements, and show that it possesses many nice features. Then we show that, with respect to this metric, the following classes of measurements can be efficiently tested: 1. the stabilizer measurements, which play a crucial role for quantum error correction; 2. the $k$-local measurements, i.e. measurements whose outcomes depend on a subsystem of at most $k$ qudits; 3. the permutation-invariant measurements, which include those measurements used in quantum data compression, state estimation and entanglement concentration. In fact, all of them can be tested with query complexity independent of the system's dimension. Furthermore, we also present an algorithm that can test any finite set of measurements. Finally, we consider the following natural question: given two black-box measurement devices, how do we estimate their distance? We give an efficient algorithm for this task, and its query complexity is also independent of the system's dimension. As a consequence, we can easily test whether two unknown measurements are identical or very different.	
\end{abstract}


\maketitle

\section{Introduction}
%importance of entangled measurements
Quantum measurements are ubiquitous and play a pivotal role in the quantum information science. Many tasks require not only the most general quantum measurements, but also the entangled ones on many particles. For example, several quantum data compression \cite{Sch95,JS94,JHH+98,HM02a,BHL06} and communication protocols \cite{BRS03} work by performing some entangled measurement related to the representation theory of symmetric group on a multi-particle state. For another instance, certain non-abelian hidden subgroup problems are reduced to the problem of efficiently performing certain joint measurement on multiple coset states 
\cite{EHK04,Kup05,BCD05,MR05,BCD06,HKK08}. In order to guarantee the success of these tasks in practice, it is crucial to make sure that the required measurements are implemented with sufficiently high fidelity. 

%hard to estimate measurements
Now imagine that someone builds a quantum measurement device and claims that it performs some measurement that we need. How do we check if this is true? Of course, we can use the quantum tomography \cite{NC00,CN97,PCZ97,LS99,Fiu01,LFC+08,FLC+09} to completely characterize the measurement implemented by the device. However, this method is very  inefficient, especially when the dimension of the system is very large. Specifically, a general quantum measurement with $k$ possible outcomes on a $D$-dimensional system is described by a collection $\{M_1,M_2,\dots,M_k\}$ of measurement operators satisfying the \textit{completeness equation}
\bea
\sum\limits_{i=1}^k M_i^{\dagger}M_i=I.
\eea 
Here $M_i$ is a bounded linear operator on the $D$-dimensional Hilbert space and it corresponds to the $i$-th outcome. In order to reconstruct the $M_i$'s, we need to determine $\Theta (k D^2)$ parameters. If the system consists of $n$ $d$-dimensional particles (or qudits), then $D=d^n$ and hence we must access the device $\Omega (k d^{2n})$ times. So this approach is practical only for systems of moderate sizes. 

%instead, we consider property testing
Given the great difficulty of fully characterizing a black-box measurement, we need to come up with a better method. Suppose our desired measurement has certain property, but the unknown measurement is shown to be very different from any measurement possessing this property, then we know that it is surely not what we want. Therefore, we can consider devising a test that separates the measurements possessing certain property from those ``far away" from them. Of course, in order to claim that two measurements are far away, we need to introduce a metric (or distance function) than quantifies their difference in the first place. We call this task \textit{property testing} of quantum measurements.

Generally speaking, property testing \cite{BLR93,RS96,GGR98} is the task of deciding whether an object has certain property or is far away from any object possessing this property, given the promise that it is one of the two cases. For example, given a boolean function $f:\{0,1\}^n \to \{0,1\}$ as an oracle (which receives the input $x \in \{0,1\}^n$ and returns the value of $f(x)$), we may want to determine whether this function is linear or far away from any linear functions, by querying the oracle as few times as possible. The property testing of classical objects, such as boolean functions and graphs, has been extensively studied in  computer science, and it plays an important role in probabilistically checkable proofs (PCP)  \cite{ALM+98}. Remarkably, many properties of boolean functions and graphs are found to be testable with very few queries. In fact, sometimes the query complexity is even independent of the object's size. The property testing of quantum objects, including quantum states \cite{HM10} and operations  \cite{Low09,MO09,HM10,FL11,SCP11,Wan11}, has been addressed only recently. It was found that many interesting classes of quantum states and operations, such as the product pure states \cite{HM10} and Clifford operations \cite{Low09,Wan11}, can also be tested with very few copies or queries.

%our results
In this paper, we initiate the study of property testing of general multi-qudit measurements. First, we introduce a distance function for quantum measurements and show that it has many nice features. In particular, if two measurements are close with respect to this metric, then they behave similarly on most input states. Then, we present efficient algorithms for testing three interesting classes of measurements: 1. the stabilizer measurements, which play a crucial role for quantum error correction \cite{Sho95,BDS+96,Ste96}; 2. the $k$-local measurements, i.e. measurements whose outcomes depend only on a subsystem of at most $k$ qudits; 3. the permutation-invariant measurements, which includes those measurements used in quantum data compression \cite{Sch95,JS94,JHH+98,HM02a,BHL06}, state estimation \cite{VLP+99,KW01,BBG+06} and entanglement concentration \cite{BBP+96,HM02b}. In fact, all of them can be tested with query complexity independent of the system's dimension. Furthermore, we also present an algorithm that can test any finite set of measurements. Finally, we consider the following natural question: given two black-box measurement devices, how do we estimate their distance? We give an efficient algorithm for this task, and its query complexity is also independent of the system's dimension. As a corollary, we can easily test whether two unknown measurements are identical or very different.

%organization
The remainder of this paper is organized as follows. In Sec.II, we introduce a metric for quantum measurements, formally describe our problem and also present several useful tools for our work. Then, in Sec. III, IV and V, we study the testing of stabilizer measurements, $k$-local measurements and permutation-invariant measurements. Next, in Sec. VI we give an algorithm that tests any finite set of measurements. After that, in Sec, VII we present an algorithm for estimating the distance between any two measurements. Finally, Sec. VIII concludes this paper.

\section{Preliminaries}
\label{sec:pre}

\subsection{A Metric for Quantum Measurements}

Consider a $D$-dimensional quantum system. Let $\HD$ be its Hilbert space, and let $\BHD$ be the set of bounded linear operators on $\HD$. Then any measurement with $k$ possible outcomes on this system can be described by $M=\{M_1,M_2,\dots,M_k\}$ where $M_i \in \BHD$ and $\sum_{i=1}^k M_i^{\dagger}M_i=I$. If we perform $M$ on the state $\rho$, then the probability of obtaining outcome $i$ is 
\bea
p(i)=\tr (M_i \rho M_i^{\dagger}),
\label{eq:dist}
\eea 
and accordingly the state of the system after the measurement becomes 
\bea
\rho_i=\dfrac{M_i \rho M_i^{\dagger}}{\tr (M_i \rho M_i^{\dagger})}.
\label{eq:post}
\eea 
It is worth noting that each $M_i$ is determined only up to a phase. Namely,  if we replace $M_i$ by $\beta_iM_i$ for some $\beta_i \in \C$, $|\beta_i|=1$, we are still describing the same measurement. One can easily see this from Eqs.(\ref{eq:dist}) and (\ref{eq:post}). In order to conveniently deal with this equivalence relation, we make the following definitions. For any $A \in \BHD$, let $[A]=\{\beta A: \beta \in \C, |\beta|=1\}$. Furthermore, for any $M=\{M_i\}_{1 \le i \le k}$, let $[M]=\{\{\beta_iM_i\}_{1 \le i \le k}: \beta_i \in \C, |\beta_i|=1\}$. Then $M$ and $N$ correspond to the same measurement if and only if $[M]=[N]$. 

For some applications, the post-measurement state is of little interest. In such cases, the measurement can be more conveniently described by a \textit{positive-operator valued measure}(POVM) $\{E_1,E_2,\dots,E_k\}$ where $E_i \coloneqq M_i^{\dagger}M_i$. But in this paper we care a lot about the post-measurement state. We would say that two measurements are identical if both the distributions of outcomes and the post-measurement states are exactly the same for the two measurements performed on arbitrary state. 

We want to be able to compare any two measurements that may have different number of outcomes. In order to  achieve this, we append infinitely many zero operators to $M=\{M_1,M_2,\dots,M_k\}$ and rewrite it as $\{M_i\}_{i \in \N}$ where $M_i=0$ for any $i>k$. In addition, we use $|M|=k$ to denote the number of nonzero $M_i$'s.

Now let $\Omega$ be the set of all measurements on the $D$-dimensional system. We want to define a metric (or distance function) $\Delta: \Omega \times \Omega \to \mathbb{R}$ that should satisfy the following natural conditions: 
for any $M=\{M_i\}_{i \in \N},N=\{N_i\}_{i \in \N},L=\{L_i\}_{i \in \N} \in \Omega$, 
\ben
\item $\Delta(M,N) \ge 0$;
\item $\Delta(M,N)=0$ if and only if $[M]=[N]$; 
\item $\Delta(M,N)=\Delta(N,M)$;
\item $\Delta(M,L) \le \Delta(M,N)+\Delta(N,L)$.
\een
Furthermore, note that a measurement $M=\{M_i\}_{i \in \N}$ on the system $A$ can be equivalently viewed as a joint measurement $M \otimes I \coloneqq \{M_i \otimes I\}_{i \in \N}$ on the system $A \otimes B$, where $B$ is any ancilla system, and similarly for $N=\{N_i\}_{i \in \N}$. Since the distance between $M$ and $N$ should be independent of whether an ancilla system is appended, or what kind of ancilla system is appended, we need
\ben[resume]
\item 
$\Delta(M,N)=\Delta(M \otimes I,N \otimes I)$,
where $I$ is the identity operation on any finite-dimensional Hilbert space.
\een
Finally, since  we want to use $\Delta$ for property testing, it needs to be normalized, i.e.
\ben[resume]
\item $\Delta(M,N) \le 1$.
\een 

Our idea is that $\Delta(M,N)$ should reflect the total differences between each pair of $M_i$ and $N_i$. First, we quantify the distance between any $A,B \in \BHD$ as:
\beba
\Delta(A,B)
&\coloneqq & \inf\limits_{A'\in [A], B' \in [B]} \dfrac{1}{\sqrt{2D}}\|A'-B'\|_F\\
&=&\inf\limits_{\theta \in [0,2\pi)} \dfrac{1}{\sqrt{2D}}\|A-e^{i\theta}B\|_F,
\eeea
where 
\bea
\|A\|_F \coloneqq \sqrt{\tr (A^{\dagger}A)}=\sqrt{\sum\limits_{i,j=1}^D |a_{i,j}|^2}
\eea 
is the Frobenius norm (or Hilbert-Schmidt norm) for $A=[a_{i,j}]$ (fixing an orthonormal basis for $\HD$, any operator $A \in \BHD$ is represented by a $D \times D$ matrix). This metric has been used in Ref. for the property testing of  unitary operations. It possesses the following nice features:
\ben[label=(\roman{*})]
\item $\Delta(A,B) \ge 0$;
\item $\Delta(A,B)=0$ if and only if 
$[A]=[B]$; 
\item $\Delta(A,B)=\Delta(B,A)$;
\item $\Delta(A,C)\le \Delta(A,B)+\Delta(B,C)$;
\item
$\Delta(A,B)=\Delta(A \otimes I,B \otimes I)$, where $I$ is the identity operation on any finite-dimensional Hilbert space.
\een
Furthermore, note that
\beba
\Delta^2(A,B)=\dfrac{1}{2D}(\ip{A}{A}+\ip{B}{B}-2|\ip{A}{B}|),
\label{eq:deltaab}
\eeea
where
\bea
\ip{A}{B}\coloneqq\tr (A^{\dagger}B)
\eea
is the Hilbert-Schmidt inner product.

Now we define the distance between $M=\{M_i\}_{i \in \N}$
and $N=\{N_i\}_{i \in \N}$ as
\beba
\Delta(M,N)&\coloneqq &
\sqrt{\sum\limits_{i \in \N} \Delta^2(M_i,N_i)}\\
&=&\sqrt{\sum\limits_{i=1}^{\max\{|M|,|N|\}} \Delta^2(M_i,N_i)}.
\label{eq:Deltadef}
\eeea

Properties (i), (ii), (iii) and (v) of $\Delta(A,B)$ immediately imply that $\Delta(M,N)$ satisfies conditions 1, 2, 3 and 5. In addition, $\Delta(M,N)$ also fulfills condition 6 because, by Eq.(\ref{eq:deltaab}) and the completeness equation, 
\beba
\Delta^2(M,N)
&=&\sum\limits_{i \in \N} \Delta^2(M_i,N_i)\\
&=&\dfrac{1}{2D} \sum\limits_{i \in \N} (\ip{M_i}{M_i}+\ip{N_i}{N_i}\\
&&-2|\ip{M_i}{N_i}|)\\
&=&1-\dfrac{1}{D} \sum\limits_{i \in \N} |\ip{M_i}{N_i}| \\
& \le & 1.
\label{eq:DeltaMN2}
\eeea
Moreover, $\Delta(M,N)$ also satisfies condition 3. To prove this, we use  property (iv) of $\Delta(A,B)$ and the following lemma:
\begin{lemma}
If $a_i$, $b_i$ and $c_i$ are non-negative numbers such that $c_i \le a_i+b_i$ for $i=1,2,\dots,k$, then 
\bea
\sqrt{\sum\limits_{i=1}^k c_i^2}
\le \sqrt{\sum\limits_{i=1}^k a_i^2}
+ \sqrt{\sum\limits_{i=1}^k b_i^2}.
\eea 
\label{lem:tri.ine}
\end{lemma}
\begin{proof}
We need to prove 
\bea
\sum\limits_{i=1}^k c_i^2
\le \sum\limits_{i=1}^k a_i^2
+ \sum\limits_{i=1}^k b_i^2
+2 \sqrt{(\sum\limits_{i=1}^k a_i^2)(\sum\limits_{i=1}^k b_i^2)}.
\eea
Since $0 \le c_i \le a_i+b_i$,
we have
\bea
\sum\limits_{i=1}^k c_i^2
\le \sum\limits_{i=1}^k (a_i+b_i)^2
=\sum\limits_{i=1}^k (a_i^2+b_i^2+2a_ib_i).
\eea
So it is sufficient to show
\bea
\sum\limits_{i=1}^k a_ib_i
\le 
\sqrt{(\sum\limits_{i=1}^k a_i^2)
(\sum\limits_{i=1}^k b_i^2)},
\eea
which is exactly the Cauchy-Schwarz inequality.
\end{proof}
Plugging $a_i=\Delta(M_i,N_i)$, $b_i=\Delta(N_i,L_i)$, $c_i=\Delta(M_i,L_i)$ for $1 \le i \le \max\{|M|,|N|,|L|\}$ into lemma \ref{lem:tri.ine}, we obtain 
\bea
\Delta(M,L) \le \Delta(M,N)+\Delta(N,L).
\eea

Now we show that $\Delta(M,N)$ reflects the \textit{average} difference between the ``behaviors" of $M$ and $N$ on a random input state. Here by ``behavior" we mean both the distribution of measurement outcomes and the post-measurement states. Specifically, let $\ket{\psi}$ be a random pure state chosen according to the normalized Haar measure. If we perform $M=\{M_i\}_{i \in \N}$ or $N=\{N_i\}_{i \in \N}$ on $\ket{\psi}$, then the unnormalized post-measurement states would be $\{M_i \ket{\psi}\}_{i \in \N}$ or $\{N_i \ket{\psi}\}_{i \in \N}$. We quantify the total difference between the two sets of states as
\beba
\sum\limits_{i \in \N} \|M_i\ket{\psi}
-N_i \ket{\psi}\|^2 
&=&
\sum\limits_{i \in \N}
\bra{\psi} (M_i^{\dagger}M_i+N_i^{\dagger}N_i\\
&&-M_i^{\dagger}N_i-N_i^{\dagger}M_i)\ket{\psi}.
\eeea
Integrating this quantity over $\ket{\psi}$ and using the fact
\beba
\dint \ketbra{\psi}{\psi} ~d \psi
=\dfrac{I}{D}
\eeea
and the completeness equation, we obtain
\beba
\dint \sum\limits_{i \in \N} \|M_i\ket{\psi}
-N_i \ket{\psi}\|^2  ~d \psi
&=& 2-\dfrac{2}{D}\sum\limits_{i \in \N} \Re\ip{M_i}{N_i}.
\eeea
By multiplying each $N_i$ by an appropriate phase (without changing the measurement being described), we can make $\ip{M_i}{N_i}$ real and non-negative. Then by Eq.(\ref{eq:DeltaMN2}), 
\beba
\dint \sum\limits_{i \in \N} \|M_i\ket{\psi}
-N_i \ket{\psi}\|^2 ~d\psi
&=&2-\dfrac{2}{D}\sum\limits_{i \in \N} |\ip {M_i}{N_i}|\\
&=&2 \Delta^2(M,N).
\label{eq:dMNP}
\eeea
Thus, if $\Delta(M,N)=\epsilon$ is small, then the expectation of $\sum_{i \in \N}\|M_i\ket{\psi}-N_i \ket{\psi}\|^2$ is $2\epsilon^2$, which implies $M_i\ket{\psi}$ and $N_i \ket{\psi}$ are close on average. Furthermore, let $p_i(\ket{\psi})=\|M_i\ket{\psi}\|^2$ and $q_i(\ket{\psi})=\|N_i\ket{\psi}\|^2$ be the probability of obtaining outcome $i$ when performing $M$ and $N$ on $\ket{\psi}$ respectively. Then by the fact
\bea
\|u-v\| \ge | \|u\|-\|v\| |~~~\forall u,v \in \CD,
\eea
we have
\beba
\sum\limits_{i \in \N}
\|M_i\ket{\psi}-N_i \ket{\psi}\|^2
&\ge &
\sum\limits_{i \in \N}
(\|M_i\ket{\psi}\|-\|N_i \ket{\psi}\|)^2\\
&=&\sum\limits_{i \in \N}
(\sqrt{p_i(\ket{\psi})}-\sqrt{q_i(\ket{\psi})})^2\\
&=&2-2F(\vp(\ket{\psi}),\vq(\ket{\psi})),
\label{eq:MPNP}
\eeea
where 
\beba
F(\vp(\ket{\psi}),\vq(\ket{\psi}))
\coloneqq \sum\limits_{i \in \N} \sqrt{p_i(\ket{\psi})q_i(\ket{\psi})}
\eeea
is the fidelity of $\vp(\ket{\psi})\coloneqq(p_i(\ket{\psi}))_{i \in \N}$ and $\vq(\ket{\psi})\coloneqq(q_i(\ket{\psi}))_{i \in \N}$. Taking the expectation of Eq.(\ref{eq:MPNP}), then we know from Eq.(\ref{eq:dMNP}) that
$F(\vec{p},\vec{q})$ is at least $1-\Delta^2(M,N)$ on average. So if $\Delta(M,N)$ is small, then $\vec{p}$ and $\vec{q}$ are  also close on average.

Note that by using the Markov inequality 
\bea
\Pr(|X|>a) \le \dfrac{\E(|X|)}{a},~~~\forall a>0,
\label{eq:Markov}
\eea
we can estimate the fraction of ``good" input states on which $M$ and $N$ behave similarly. Specifically, let $\Delta(M,N)=\epsilon$. Then setting $X=\sum_{i \in \N}\|M_i\ket{\psi}-N_i \ket{\psi}\|^2$ and $a=10\epsilon^2$ in Eq.(\ref{eq:Markov}) yields 
\bea
\Pr[\sum_{i \in \N}\|M_i\ket{\psi}-N_i \ket{\psi}\|^2 \ge 10\epsilon^2] \le \dfrac{1}{5}.
\eea
So for at least $4/5$ fraction of $\ket{\psi}$'s, $M\ket{\psi}$ and $N\ket{\psi}$ are close, provided that $\epsilon$ is small enough. Similarly, setting $X=1-F(\vp(\ket{\psi}),\vq(\ket{\psi}))$ and $a=10\epsilon^2$ in Eq.(\ref{eq:Markov}) yields
\bea
\Pr[1-F(\vp(\ket{\psi}),\vq(\ket{\psi})) \le 10\epsilon^2] \le \dfrac{1}{5}.
\eea
So for at least $4/5$ fraction of $\ket{\psi}$'s, $\vp(\ket{\psi})$ and $\vq(\ket{\psi})$ are close, provided that $\epsilon$ is small enough.

\subsection{Property Testing of Quantum Measurements}
The task of property testing can be formally described as follows. Suppose $\Omega$ is a class of mathematical objects equipped with a metric $\Delta$. A property is just a subset $\S \subset \Omega$. For any $A \in \Omega$, if $A \in \S$, then we say that $A$ has property $S$; otherwise, if $\Delta(A,\S) \ge \epsilon$, i.e. $\Delta(A,B)\ge \epsilon$
for any $B \in \S$, then we say that $A$ is $\epsilon$-far from property $\S$. An algorithm $\epsilon$-tests property $\S$ with query complexity $q(|\Omega|,\epsilon)$ if on any input $A \in \Omega$, it makes at most $q(|\Omega|,\epsilon)$ queries to $A$ and behaves as follows:
\bit
\item if $A$ has property $\S$, then the algorithm accepts $A$ with probability at least $2/3$;
\item if $A$ is $\epsilon$-far from property $\S$, then the algorithm accepts $A$ with probability at most $1/3$.
\eit
Note that by repeating this algorithm many times and choosing the majority answer, we can exponentially decrease the probability of making an erroneous decision. So the completeness error $1-2/3=1/3$ here and soundness error $1/3$ here can be replaced by arbitrarily small constants.

In this paper, we study the property testing of quantum measurements on finite-dimensional systems with respect to the metric $\Delta$ defined by Eq.(\ref{eq:Deltadef}). Specifically, fix a $D$-dimensional system, and let $\Omega$ be the set of all measurements on this system. Suppose $\S \subset \Omega$ is the property to be tested. Our input is a black box performing some unknown measurement $M \in \Omega$. We can access it by preparing the $D$-dimensional system (denoted by $A$) in some known state $\rho_A$, applying $M$ on the state,  and obtaining an outcome $i$ as well as the post-measurement state $M_i \rho_{A} M_i^{\dagger}$ (up to a normalization). Moreover, we can also introduce an ancilla system $B$ and prepare the joint system $A \otimes B$ in some entangled state $\rho_{AB}$, then apply $M$ on subsystem $A$ of this state, and get an outcome $i$ as well as the post-measurement state $(M_i \otimes I) \rho_{AB} (M_i^{\dagger} \otimes I)$ (up to a normalization). An algorithm $\epsilon$-tests $\S$ with query complexity $q(D,\epsilon)$ if for any input $M \in \Omega$, it accesses the black box in the above way at most $q(D,\epsilon)$ times and makes a correct decision about whether $M \in \S$ or $\Delta(M,\S)\ge \epsilon$ with probability at least $2/3$. Note that we allow the algorithm to extract useful information about $M$  from both the outcome statistics and post-measurement states. Our goal is to devise such a testing algorithm with the minimal query complexity. In addition, we also prefer this algorithm to be efficiently implementable. Namely, we want its time complexity to be polynomial in $\log(D)$ and $1/\epsilon$, assuming one query to the black box takes a unit time.

\subsection{Useful Tools}
\label{sec:tool}
The following tools will be very useful for our work.

\subsubsection{Probability Theory}
Due to the probabilistic nature of quantum measurements, our testing algorithms heavily depend on analyzing the outcome statistics. So the following results from probability theory will be very helpful.

Suppose $\vp=(p_i)_{i \in \N}$ and 
$\vq=(q_i)_{i \in \N}$ are two probability distributions over $\N$. Let
\bea
D(\vp,\vq)\coloneqq \dfrac{1}{2}\sum\limits_{i \in \N} |p_i-q_i|
\eea
be the variational distance of $\vp$ and $\vq$, and let 
\bea
F(\vp,\vq)\coloneqq \sum\limits_{i \in \N} \sqrt{p_iq_i}.
\eea
be the fidelity of $\vp$ and $\vq$.
It is easy to see that $0\le \Delta(\vp,\vq) ,F(\vp,\vq)\le 1$. In addition, they satisfy the following relation:
\bea
1-F(\vp,\vq) \le D(\vp,\vq) \le \sqrt{1-F^2(\vp,\vq)}.
\label{eq:st&f}
\eea

Meanwhile, the Chernoff-Hoeffding bound \cite{Che52,Hoe63} is another important and useful result from probability theory. It gives an exponentially decreasing bounds on tail distributions of sums of independent random variables. Specifically, let $X_1,X_2,\dots,X_n \in [0,1]$ be independent random variables and let $S=\sum_{i=1}^n X_i$. Then for any $0<\epsilon<\E[S]/n$,
\bea
\Pr \left[\left|\dfrac{S-\E[S]}{n} \right|>\epsilon \right] \le 2 e^{- 2n\epsilon^2}.
\label{eq:chernoff}
\eea
In particular, if the $X_i$'s are independent and identically distributed (i.i.d.) as $X$, then 
\bea
\Pr\left[\left|\dfrac{S}{n}-\E[X]\right| > \epsilon \right] \le 2 e^{- 2n\epsilon^2}.
\label{eq:chernoff2}
\eea
By choosing $n=O(\log(1/{\delta})/{\epsilon^2})$ we can make the right-hand side smaller than $\delta$. Then $\E[X]$ can be approximated by $S/n$ with precision $\epsilon$ and confidence $1-\delta$. 

Now imagine that $X$ is a $\{0,1\}$-valued random variable indicating whether a particular outcome occurs in an experiment. Namely, $X=1$ if a this outcome occurs, and $X=0$ otherwise. Then $\E[X]$ is the probability for this outcome in this experiment. By Eq.(\ref{eq:chernoff2}), we can approximate this probability with precision $\epsilon$ and confidence $1-\delta$ by repeating the experiment $O(\log(1/{\delta})/{\epsilon^2})$ times.

\subsubsection{The Choi-Jamio\l kowski Isomorphism}

The Choi-Jamio\l kowski isomorphism \cite{Jam72,Cho75} states that there is a duality between quantum operations and quantum states. Specifically, let
\bea
\ME \coloneqq \dfrac{1}{\sqrt{D}}\sum\limits_{i=1}^{D}\ket{i}\ket{i}
\eea
be the $D$-dimensional maximally entangled state. For any $A \in \BHD$, define
\beba
\ket{v(A)}\coloneqq (A \otimes I)\ME,
\eeea
where $A$ acts on the first subsystem. Then for any $A$, $B$,
\beba
\bracket{v(A)}{v(B)}&=&\dfrac{1}{D}\ip{A}{B}.\\
\label{eq:isoip}
\eeea
Namely, the inner product between $\ket{v(A)}$ and $\ket{v(B)}$ is proportional to the Hilbert-Schimdt product of $A$ and $B$. In particular, if $A=B$, then we get 
\bea
p(A) \coloneqq \|\ket{v(A)}\|^2=\dfrac{1}{D}\|A\|_F^2.
\eea
So $\ket{v(A)}$ is not normalized unless $\|A\|_F=\sqrt{D}$. We use $\ket{\tilde{v}(A)}$ to denote the normalized version
of $\ket{v(A)}$, i.e.
\bea
\ket{\tilde{v}(A)}\coloneqq\dfrac{\ket{v(A)}}{\|\ket{v(A)}\|}.
\eea

Now suppose that we perform a measurement $M=\{M_i\}_{i \in \N}$ on the first subsystem of $\ME$, then the unnormalized post-measurement states would be $\{\ket{v(M_i)}\}_{i \in \N}$. Namely, the outcome $i$ occurs with probability 
\beba
p(M_i) = \|\ket{v(M_i)}\|^2=\dfrac{1}{D}\|M_i\|_F^2, 
\eeea
and the corresponding post-measurement state is 
\bea
\ket{\tilde{v}(M_i)}=\dfrac{1}{\sqrt{p(M_i)}}\ket{v(M_i)}.
\eea 
This property will be crucial for every testing algorithm given in this paper. 

\subsubsection{Pauli Decomposition}
Let
\beba
\sigma_{00}=I=\bpm
1 & 0\\
0 & 1\\
\epm,&
\sigma_{10}=X=\bpm
0 & 1\\
1 & 0\\
\epm,\\
\sigma_{01}=Z=\bpm
1 & 0\\
0 & -1\\
\epm &
\sigma_{11}=Y=\bpm
0 & -i\\
i & 0\\
\epm
\eeea
be the Pauli operators. For any $\vx=(x_1,x_2,\dots,x_n)$,  $\vz=(z_1,z_2,\dots,z_n) \in \Ztn$, let
\beba
\sigma_{\vx,\vz}=\sigma_{x_1,z_1} \otimes \sigma_{x_2,z_2} \otimes \dots \otimes \sigma_{x_n,z_n}.
\label{eq:sigmaxz}
\eeea
Then for any $\va,\vb,\vc,\vd \in \Ztn$, we have
\bea
[\sigma_{\va,\vb}\sigma_{\vc,\vd}]
=[\sigma_{\va \oplus \vc,\vb \oplus \vd}].
\eea
Here `$\oplus$' denotes the bitwise addition modulo $2$.
So there exists $\beta_{\va,\vb,\vc,\vd} \in \C$ with $|\beta_{\va,\vb,\vc,\vd}|=1$ such that
\bea
\sigma_{\va,\vb}\sigma_{\vc,\vd}
=\beta_{\va,\vb,\vc,\vd}\sigma_{\va \oplus \vc,\vb \oplus \vd}.
\eea
Moreover, $\{\sigma_{\vx,\vz}\}_{\vx,\vz \in \Ztn}$
form an orthogonal basis for $\BHtn$ with respect to the Hilbert-Schmidt product. So for any $A \in \BHtn$ we can write it as 
\beba
A = \sum\limits_{\vx,\vz \in \Ztn}{\mu_{\vx,\vz}(A)\sigma_{\vx,\vz}},
\eeea
where 
\beba
\mu_{\vx,\vz}(A)\coloneqq \dfrac{1}{2^n}
\ip{\sigma_{\vx,\vz}}{A}.
\eeea
Note that
\bea
\|A\|_F^2=2^n \sum\limits_{\vx,\vz \in \Ztn}|\mu_{\vx,\vz}(A)|^2.
\eea 
Furthermore, by the Choi-Jamio\l kowski isomorphism,
\beba
\ket{v(A)}=\sum\limits_{\vx,\vz \in \Ztn}
\mu_{\vx,\vz}(A) \ket{v(\sigma_{\vx,\vz})},
\eeea
where $\{\ket{v(\sigma_{\vx,\vz})}\}_{\vx,\vz \in \Ztn}$ form an orthonormal basis for $\Htn$. So 
\beba
p(A)=\|\ket{v(A)}\|^2=\sum\limits_{\vx,\vz \in \Ztn}
|\mu_{\vx,\vz}(A)|^2.
\eeea
It follows that
\beba
\ket{\tilde{v}(A)}&=& \dfrac{1}{\sqrt{p(A)}} \ket{v(A)} \\
&=& \sum\limits_{\vx,\vz \in \Ztn}
\dfrac{\mu_{\vx,\vz}(A)}{\sqrt{p(A)}} \ket{v(\sigma_{\vx,\vz})}.
\eeea
So if we measure $\ket{\tilde{v}(A)}$ in the basis $\{\ket{v(\sigma_{\vx,\vz})}\}_{\vx,\vz \in \Ztn}$, then the outcome $(\vx,\vz)$ occurs with probability 
\beba
q_{\vx,\vz}(A) \coloneqq \dfrac{|\mu_{\vx,\vz}(A)|^2}{p(A)}. 
\eeea

Now consider any measurement $M=\{M_i\}_{i \in \N}$ on an 
$n$-qubit system. By plugging 
$M_i=\sum_{\vx,\vz \in \Ztn} \mu_{\vx,\vz}(M_i)\sigma_{\vx,\vz}$
into $\sum_{i \in \N}M_i^{\dagger}M_i=I$, we obtain
\beba
\sum\limits_{i \in \N} 
\sum\limits_{\vx,\vz \in \Ztn}
|\mu_{\vx,\vz}(M_i)|^2 
&=& 1.
\label{eq:pau.sum.1}
\eeea
Furthermore, if we measure $\ket{\tilde{v}(M_i)}$ in the basis
$\{\ket{v(\sigma_{\vx,\vz})}\}_{\vx,\vz \in \Ztn}$, then the outcome $(\vx,\vz)$ occurs with probability
\bea
q_{\vx,\vz}(M_i)=\dfrac{|\mu_{\vx,\vz}(M_i)|^2}{p(M_i)}
\eea
These properties will be very useful for the testing of the stabilizer and $k$-local measurements. 

Finally, the above facts can be straightforwardly generalized to the qudit case. We mainly need to replace the Pauli operators by their higher-dimensional analogues
\beba
\sigma_{x,z}&=&\sum\limits_{j=0}^{d-1} \omega^{jz}\ketbra{j \oplus x}{j},
\label{eq:dpauli}
\eeea
for $x,z \in \Zd$. Here $\omega=e^{i2\pi/d}$ and `$\oplus$' denotes the addition modulo $d$. The reader can easily work out the rest of the details.

\section{Testing the Stabilizer Measurements}
\label{sec:stabilizer}
Equipped with the notations introduced above, we begin with the testing of the stabilizer measurements. This kind of measurements have played a crucial role in quantum error correction \cite{Sho95,BDS+96,Ste96}. They are closely related to the stabilizer codes \cite{Got97}, which are defined by choosing a set of commuting operators from the Pauli group on $n$ qubits. Specifically, let $\Pn\coloneqq\{\pm \sigma_{\vx,\vz}: \vx,\vz \in \Ztn\}$.
Then any operator $g \in \Pn$ has eigenvalues $+1$ and $-1$, and its (+1)-eigenspace and (-1)-eigenspace are both $2^{n-1}$-dimensional. Now if we choose a set of commuting operators $g_1,g_2,\dots,g_k \in \Pn$ such that $-I \not\in \langle g_1,g_2,\dots,g_k\rangle$, then the simultaneous (+1)-eigenspace of $g_1,g_2,\dots,g_k$ defines the coding space of a stabilizer code. The error-correction for this code is quite simple: we simply measure each of $g_1,g_2,\dots,g_k$, and from the outcomes we can determine the error, and then we perform the appropriate recovery operation.

A stabilizer measurement is the projective measurement onto the $(\pm 1)$-eigenspaces of $\sigma_{\va,\vb}$ for any $\va,\vb \in \Ztn$. Formally, it is denoted by $P(\va,\vb)=\{P_1(\va,\vb),P_2(\va,\vb)\}$ where
\bea
P_{1}(\va,\vb)&\coloneqq &\dfrac{I+\sigma_{\va,\vb}}{2}, \\
P_{2}(\va,\vb)&\coloneqq &\dfrac{I-\sigma_{\va,\vb}}{2}.
\eea
So $M=\{M_i\}_{i \in \N}$ is a stabilizer measurement if and only if there exist $\va,\vb \in \Ztn$ such that
\beba
\mu_{\vzero,\vzero}(M_{1})=
\mu_{\va,\vb}(M_{1})=
\mu_{\vzero,\vzero}(M_{2})=
-\mu_{\va,\vb}(M_{2})=\dfrac{1}{2},\\
\mu_{\vx,\vy}(M_i)=0,~~\forall i \ge 3 \textrm{~or~}
(\vx,\vy) \neq (\vzero,\vzero), (\va,\vb).
\eeea
The following lemma says that if the above condition is approximately fulfilled, then $M=\{M_i\}_{i \in \N}$ is close to a stabilizer measurement.
\begin{lemma}
For any $0 \le \gamma,\delta \le 1/4$
and $\va,\vb \in \Ztn$, if a measurement $M=\{M_i\}_{i \in \N}$ on $n$ qubits satisfies 
\bea
|\mu_{\vzero,\vzero}(M_{1})|,|\mu_{\va,\vb}(M_{1})|
\in \left[\sqrt{\dfrac{1}{4}-\gamma}, \sqrt{\dfrac{1}{4}+\gamma}\right],
\label{eq:condm11}\\
|\mu_{\vzero,\vzero}(M_{2})|,|\mu_{\va,\vb}(M_{2})|
\in \left[\sqrt{\dfrac{1}{4}-\gamma}, \sqrt{\dfrac{1}{4}+\gamma}\right],
\label{eq:condm21}
\eea
and
\bea
\Re(\mu_{\vzero,\vzero}^*(M_{1})
\mu_{\va,\vb}(M_{1})) \ge \dfrac{1}{4}-\delta, 
\label{eq:condm12}\\
\Re(\mu_{\vzero,\vzero}^*(M_{2})
\mu_{\va,\vb}(M_{2})) \le -\dfrac{1}{4}+\delta, 
\label{eq:condm22}
\eea
then 
\bea
\Delta(M,P(\va,\vb)) \le \sqrt{8\gamma+2\delta}.
\eea
\label{lem:pauliclose}
\end{lemma}
\begin{proof} See Appendix \ref{apd:pauliclose}.
\end{proof}

With the help of lemma \ref{lem:pauliclose}, we obtain the following result on testing the stabilizer measurements.

\begin{theorem}
The stabilizer measurements on $n$ qubits can be $\epsilon$-tested with query complexity $O(1/\epsilon^4)$. Furthermore, the testing algorithm can be efficiently implemented.
\label{thm:stabilizer}
\end{theorem}
\begin{proof}
Given a black box performing some unknown measurement $M=\{M_i\}_{i \in \N}$ on $n$ qubits, we run the following test on it: \\

\begin{algorithm}[H]
\caption{~~~Testing the stabilizer measurements}
\ben
\item Let $D=2^n$, $L=\left\lceil \dfrac{20000}{\epsilon^4}\right\rceil$, $N=\left\lfloor \left(\dfrac{1}{2}-\dfrac{\epsilon^2}{64}\right)L \right\rfloor$,
$T= \left\lceil 0.99 N  \right\rceil$, 
$W= \left\lceil \dfrac{12}{\epsilon^2}  \right\rceil$.
\item Prepare $L$ copies of 
$$~~~~~~~\ME=\dfrac{1}{\sqrt{D}}\sum\limits_{(i_1,i_2,\dots,i_n)\in \Ztn}\ket{i_1,i_2,\dots,i_n}\ket{i_1,i_2,\dots,i_n}.$$
\item Perform $M=\{M_i\}_{i \in \N}$ on the first subsystem of each copy of $\ME$. If we obtain only outcomes $1$ and $2$ in these measurements, and the fraction of outcome $1$ is between $\dfrac{1}{2}-\dfrac{\epsilon^2}{64}$ and
$\dfrac{1}{2}+\dfrac{\epsilon^2}{64}$, then continue; otherwise, reject $M$ and quit; 
\item Now we should have at least $N$ copies of $\ket{\tilde{v}(M_{1})}$ and 
$\ket{\tilde{v}(M_{2})}$ which are the post-measurement states corresponding to the outcome $1$ and $2$ respectively.
\item Choose $T$ copies of $\ket{\tilde{v}(M_{1})}$ and 
$\ket{\tilde{v}(M_{2})}$, and measure each of them in the basis $\{\ket{v(\sigma_{\vx,\vz})}\}_{\vx,\vz \in \Ztn}$. If we obtain only two different outcomes $(\textbf{0},\textbf{0})$ and $(\va,\vb)$ for some $\va,\vb \in \Ztn$ in these measurements, and the fraction of outcome $(\textbf{0},\textbf{0})$ is between $\dfrac{1}{2}-\dfrac{\epsilon^2}{64}$ and $\dfrac{1}{2}+\dfrac{\epsilon^2}{64}$ for both those measurements on $\ket{\tilde{v}(M_{1})}$ and those measurements on $\ket{\tilde{v}(M_{2})}$, then continue; otherwise, reject $M$ and quit. 
\item Select another $W$ copies of $\ket{\tilde{v}(M_{1})}$
and $\ket{\tilde{v}(M_{2})}$. For each copy, do the following: 
\ben
\item For $j=1,2,\dots,n$, if $a_j=b_j=0$, then set $s_j=1$; otherwise, measure the $j$-th qubit of the first subsystem in the $\sigma_{a_j,b_j}$ basis and set $s_j$ to be the outcome $+1$ or $-1$. 
\item If the measured state is $\ket{\tilde{v}(M_{1})}$ and $\prod\limits_{i=1}^n s_j=-1$,
or the measured state is $\ket{\tilde{v}(M_{2})}$ and
$\prod\limits_{i=1}^n s_j=1$, then reject $M$ and quit.
\een
\item Now $M$ has passed all the above tests. Accept $M$. 
\een 
\label{alg:stabilizer}
\end{algorithm}

For correctness, we need to prove: \\
(1) if $M$ is a stabilizer measurement, then this algorithm accepts $M$ with probability at least $2/3$;\\ 
(2) on the other hand, if this algorithm accepts $M$ with probability at least $1/3$, then $M$ is $\epsilon$-close to some stabilizer measurement. (Taking the contrapositive, we get that if $M$ is $\epsilon$-far away from any stabilizer measurement, then this algorithm accepts it with probability at most $1/3$.)

Before proving the two statements, observe that:
\ben
\item
In step 3, for $i=1,2$, the outcome $i$ should occur with probability $p(M_i)$. So, by Eq.(\ref{eq:chernoff2}) and our choice of $L$, the fraction of outcome $i$ is $\epsilon^2/64$-close to $p(M_i)$ with probability at least $0.95$.
\item
In step 5, for $i=1,2$, if the measured state is $\ket{\tilde{v}(M_i)}$, then the outcome $(\vzero,\vzero)$ (or $(\va,\vb)$) should occur with probability $q_{\vzero,\vzero}(M_{i})$ (or $q_{\va,\vb}(M_{i})$). So, by Eq.(\ref{eq:chernoff2}) and our choice of $T$, the fraction of outcome $(\vzero,\vzero)$ (or $(\va,\vb)$) is $\epsilon^2/64$-close to $q_{\vzero,\vzero}(M_{i})$ (or $q_{\va,\vb}(M_{i})$) with probability at least $0.95$.
\item 
In step 6, for $i=1,2$, if the measured state is $\ket{\tilde{v}(M_i)}$, then the $s_j$'s depend only on the reduced state of the first subsystem of $\ket{\tilde{v}(M_i)}$, which is 
\bea
\rho_i=\dfrac{M_iM_i^{\dagger}}{\tr (M_i^{\dagger}M_i)}.
\eea
Then it is easy to see
\beba
\Pr[\prod\limits_{j=1}^n s_j=1 | \rho_1]
&=&\tr (P_1(\va,\vb)\rho_1)\\
&=&\dfrac{\tr (P_1(\va,\vb)M_1M_1^{\dagger})}{\tr (M_1^{\dagger}M_1)},\\
\label{eq:prsjrho1}
\eeea
\beba
\Pr[\prod\limits_{j=1}^n s_j=-1 | \rho_2]
&=&\tr (P_2(\va,\vb)\rho_2)\\
&=&\dfrac{\tr (P_2(\va,\vb)M_2M_2^{\dagger})}{\tr (M_2^{\dagger}M_2)}.
\label{eq:prsjrho2}
\eeea
If $M$ passes the test with probability at least $1/3$, then we must have
\bea
\Pr[\prod\limits_{j=1}^n s_j=1 | \rho_1] \ge 1-\dfrac{\epsilon^2}{4},
\label{eq:prsj1}\\
\Pr[\prod\limits_{j=1}^n s_j=-1 | \rho_2] \ge 1-\dfrac{\epsilon^2}{4},
\label{eq:prsj2}
\eea 
because, if otherwise, step 6 would reject $M$ with probability at least $1-(1-\epsilon^2/4)^W \ge 1-e^{-\epsilon^2W/4} \ge 0.9$ by our choice of $W$.
\een

Now let us prove statement (1). Suppose $M=P(\va,\vb)$ for some $\va$, $\vb$ $\in \Ztn$. Since 
\bea
p(P_{1}(\va,\vb))=p(P_{2}(\va,\vb))=1/2,
\eea
by observation 1 step 3 rejects $P(\va,\vb)$ with probability at most $0.1$. Next, since
\bea
q_{\vzero,\vzero}(P_{1}(\va,\vb))
=q_{\va,\vb}(P_{1}(\va,\vb))
=\dfrac{1}{2}, \\
q_{\vzero,\vzero}(P_{2}(\va,\vb))
=q_{\va,\vb}(P_{2}(\va,\vb))
=\dfrac{1}{2}, 
\eea
by observation 2 step 5 rejects $P(\va,\vb)$ with probability at most $0.1$. Then by observation 3 step 6 never rejects $P(\va,\vb)$. So, overall, algorithm \ref{alg:stabilizer} accepts $P(\va,\vb)$ with probability at least $0.8$.

Next, we prove statement (2). Suppose $M$ is accepted by algorithm \ref{alg:stabilizer} with probability at least $1/3$. Then we claim that: \\
(i) $M$ satisfies Eqs.(\ref{eq:condm11}) and (\ref{eq:condm21}) with $\gamma=\epsilon^2/16$; \\
(ii) $M$ satisfies Eqs.(\ref{eq:condm12}) and (\ref{eq:condm22}) with
$\delta=\epsilon^2/4$. \\
Then, by lemma \ref{lem:pauliclose}, $M$ is $\epsilon$-close to $P(\va,\vb)$. 

To prove (i), it is sufficient to show
\bea
p(M_1),p(M_2) \in \left[\dfrac{1}{2}-\dfrac{\epsilon^2}{32}, \dfrac{1}{2}+\dfrac{\epsilon^2}{32}\right]
\label{eq:pm1pm2}
\eea
and
\bea
q_{\vzero,\vzero}(M_1),
q_{\va,\vb}(M_1) \in 
\left[\dfrac{1}{2}-\dfrac{\epsilon^2}{32}, \dfrac{1}{2}+\dfrac{\epsilon^2}{32}\right],
\label{eq:qm1}\\
q_{\vzero,\vzero}(M_2),
q_{\va,\vb}(M_2) 
\in 
\left[\dfrac{1}{2}-\dfrac{\epsilon^2}{32}, \dfrac{1}{2}+\dfrac{\epsilon^2}{32}\right],
\label{eq:qm2}
\eea
because if they are true, then
\beba
|\mu_{\vx,\vz}(M_i)|^2&=&p(M_i)q_{\vx,\vz}(M_i) \\
&\in &\left[\left(\dfrac{1}{2}-\dfrac{\epsilon^2}{32}\right)^2,\left(\dfrac{1}{2}+\dfrac{\epsilon^2}{32}\right)^2\right]\\
&\subset & 
\left[\dfrac{1}{4}-\dfrac{\epsilon^2}{16}, \dfrac{1}{4}+\dfrac{\epsilon^2}{16}\right],\\
&&\forall i=1,2,~\forall (\vx,\vz)=(\vzero,\vzero),(\va,\vb).
\label{eq:muvxvzmj}
\eeea 
Eq.(\ref{eq:pm1pm2}) holds because, if otherwise, then by observation 1 step 3 would reject $M$ with probability at least $0.95$, contradicting our assumption. Similarly, Eqs.(\ref{eq:qm1})
and (\ref{eq:qm2}) also hold because, if otherwise, then 
by observation 2 step 5 would reject $M$ with probability at least $0.95$, also contradicting our assumption.

Now it only remains to prove (ii). Plugging 
\beq
M_i=\sum\limits_{\vx,\vz \in \Ztn} \mu_{\vx,\vz}(M_i)\sigma_{\vx,\vz}
\eeq 
into Eqs.(\ref{eq:prsjrho1}) and (\ref{eq:prsjrho2}) gives
\beba
\Pr[\prod\limits_{j=1}^n s_j=1 | \rho_1]
&=&\dfrac{1}{2}
+\dfrac{2 \Re(\mu_{\vzero,\vzero}^*(M_1)
\mu_{\va,\vb}(M_1))+\lambda_1}{2p(M_1)},
\label{eq:prsj3}
\eeea
\beba
\Pr[\prod\limits_{j=1}^n s_j=-1 | \rho_2]
&=&\dfrac{1}{2}
-\dfrac{2 \Re(\mu_{\vzero,\vzero}^*(M_2)
\mu_{\va,\vb}(M_2))+\lambda_2}{2p(M_2)},
\label{eq:prsj4}
\eeea
where
\beba
\lambda_i \coloneqq \sum\limits_{(\vx,\vz)\neq (\vzero,\vzero),
(\va,\vb)} \beta_{\vx,\vz,\vx\oplus \va,\vz \oplus \vb} \mu_{\vx,\vz}(M_i)
\mu_{\vx\oplus \va,\vz \oplus \vb}^*(M_i)
\eeea
for $i=1,2$. Note that
\beba
|\lambda_i|
& \le &
\sum\limits_{(\vx,\vz)\neq (\vzero,\vzero),
(\va,\vb)} |\mu_{\vx,\vz}(M_i)|
\cdot |\mu_{\vx\oplus \va,\vz \oplus \vb}(M_i)|\\
&\le & 
\sum\limits_{(\vx,\vz)\neq (\vzero,\vzero),
(\va,\vb)} |\mu_{\vx,\vz}(M_i)|^2 \\
&=& p(M_i)-|\mu_{\vzero,\vzero}(M_i)|^2
-|\mu_{\va,\vb}(M_i)|^2 \\
&\le & \left(\dfrac{1}{2}+\dfrac{\epsilon^2}{32}\right)
-2\cdot \left(\dfrac{1}{4}-\dfrac{\epsilon^2}{16}\right) \\
&\le & \dfrac{3\epsilon^2}{16},
\label{eq:lambdaj}
\eeea
where in the second last step we use Eqs.(\ref{eq:pm1pm2}) and (\ref{eq:muvxvzmj}). Now by Eqs.(\ref{eq:prsj1}),
(\ref{eq:pm1pm2}), 
(\ref{eq:prsj3}) and (\ref{eq:lambdaj}), we obtain
\bea
\Re(\mu_{\vzero,\vzero}^*(M_1)
\mu_{\va,\vb}(M_1)) \ge \dfrac{1}{4}-\dfrac{\epsilon^2}{4},
\eea
Similarly, by Eqs.(\ref{eq:prsj2}), (\ref{eq:pm1pm2}),  (\ref{eq:prsj4}) and (\ref{eq:lambdaj}), we get
\bea
\Re(\mu_{\vzero,\vzero}^*(M_2)
\mu_{\va,\vb}(M_2)) \le -\dfrac{1}{4}+\dfrac{\epsilon^2}{4},
\eea
as desired.

Finally, algorithm \ref{alg:stabilizer} has query complexity $O(1/\epsilon^4)$. Moreover, besides querying the black box, it requires $O(n/\epsilon^4)$ quantum operations including: 1. preparing the state $\ME$, which is equivalent to preparing the $j$-th and $(n+j)$-th qubits in the Bell state for $j=1,2,\dots,n$; 2. measuring a $2n$-qubit state in the basis $\{\ket{v(\sigma_{\vx,\vz})}\}_{\vx,\vz \in \Ztn}$, which is equivalent to measuring the $j$-th and $(n+j)$-th qubit in the Bell basis for $j=1,2,\dots,n$; 3. measuring any Pauli operator on a qubit. In addition, the classical processing is easy. So algorithm \ref{alg:stabilizer} can be efficiently implemented. 
\end{proof}

The stabilizer code has been generalized to the qudit case where $d$ is any prime number. So we can also consider testing the stabilizer measurements on $n$ qudits, which are the projective measurements onto the $d$ eigenspaces of $\sigma_{\vx,\vz}$ for any $\vx,\vz \in \Zdn$. Algorithm \ref{alg:stabilizer} can be straightforwardly generalized to test these measurements. We basically follow the same pattern, except that now we have $d$ equally likely outcomes instead of two. The reader can easily work out the rest of the details. This generalized algorithm still has query complexity $O(1/\epsilon^4)$, and still can be efficiently implemented.

Finally, note that algorithm 1 does not only test whether $M$ is close to a stabilizer measurement, but also learns which stabilizer measurement $M$ is closest to. This is fortunate, but not generic. Usually learning an object
is much harder than testing it, as we need to acquire more structural information about the object.

\section{Testing the $k$-Local Measurements}
\label{sec:klocal} 
Given a measurement on an $n$-qudit system, it is natural ask whether this measurement truly involves all of the $n$ qudits. Namely, its outcome might depends only on a small subsystem. Specifically, Let $[n]\coloneqq\{1,2,\dots,n\}$. For any $T=\{j_1,j_2,\dots,j_m\} \subseteq [n]$, we  use $T$ to denote the subsystem consisting of the $j_1$-th, $j_2$-th, $\dots$, $j_m$-th qubits. We say that a measurement $M=\{M_i\}_{i \in \N}$ is \textit{$k$-local} if there exists some $T \subseteq [n]$ with $|T|=k$ such that
\bea
M_i=\tilde{M}_{i|T} \otimes I_{|T^c}
\label{eq:klocaldef}
\eea
for any $i \in \N$.  Here $\tilde{M}_i$ acts on the subsystem $T$ and $I$ acts on the complementary subsystem $T^c$.  
In this section, we will study the testing of these $k$-local measurements.

For convenience, we introduce the following notations. For any $\vx=(x_1,x_2,\dots,x_n), \vz=(z_1,z_2,\dots,z_n) \in \Zdn$, let 
\bea
\supp(\vx,\vz)=\{i \in [n]: x_i \neq 0~\textrm{or}~z_i \neq 0\}
\eea
be the support of $(\vx,\vz)$. Then let
\bea
\Gamma_T = \{(\vx,\vz): \vx,\vz \in \Zdn, ~\supp(\vx,\vz) \subseteq T\}.
\eea
Then  for any $A \in \BHdn$, we can write it as
\beba
A&=&f_T(A)+g_T(A),
\eeea
where
\bea
f_T(A) &\coloneqq & \sum\limits_{(\vx,\vz)\in \Gamma_T} \mu_{\vx,\vz}(A) \sigma_{\vx,\vz},\\
g_T(A) &\coloneqq &\sum\limits_{(\vx,\vz) \not \in \Gamma_T} \mu_{\vx,\vz}(A) \sigma_{\vx,\vz}.
\eea
While $f_T(A)$ acts non-trivially only on the subsystem $T$, $g_T(A)$  does not. Also, $f_T(A)$ and $g_T(A)$ are orthogonal with respect to the Hilbert-Schmidt product. Therefore,
\beba
\|A\|_F^2=\|f_T(A)\|_F^2+\|g_T(A)\|_F^2.
\eeea

With the above notations, a measurement $M=\{M_i\}_{i \in \N}$ is $k$-local if and only if there exists some $T \subseteq [n]$ with $|T|=k$ such that $M_i=f_T(M_i)$ for any $i \in \N$. The following lemma says that if this condition is approximately fulfilled, then $M$ is close to a $k$-local measurement:

\begin{lemma}
For any $0<\delta<1$, if a measurement $M=\{M_i\}_{i\in \N}$ on $n$ qudits satisfies
\bea
\sum\limits_{i\in \N} \|f_T(M_i)\|_F^2
\ge D(1-\delta^2),
\label{eq:sumgtmi}
\eea
where $D=d^n$, then it is $\delta$-close to a $|T|$-local measurement.
\label{lem:klocal}
\end{lemma}
\begin{proof}
See Appendix \ref{apd:lemklocal}.
\end{proof}

With the help of lemma \ref{lem:klocal}, we have the following result on testing the $k$-local measurements.

\begin{theorem}
The $k$-local measurements on $n$ qudits can be $\epsilon$-tested with query
complexity $O(k \log (k/\epsilon)/ \epsilon^2)$. Furthermore, the testing algorithm can be efficiently implemented.
\end{theorem}
\begin{proof}
Given a black box performing some unknown measurement $M=\{M_i\}_{i \in \N}$ on $n$ qudits, we run the following test on it: 

\begin{algorithm}[H]
\caption{~~~Testing the $k$-local measurements}
\ben
\item Let $D=d^n$, $L=\left\lceil \dfrac{1200k}{\epsilon^2}\left[\ln\left(\dfrac{k}{\epsilon}\right)+1\right]\right\rceil$.
\item Prepare $L$ copies of $$~~~~~~~~\ME=\dfrac{1}{\sqrt{D}}\sum\limits_{(i_1,i_2,\dots,i_n)\in \Zdn}\ket{i_1,i_2,\dots,i_n}\ket{i_1,i_2,\dots,i_n}.
$$
\item For $j=1,2,\dots, L$, perform $M=\{M_i\}_{i \in \N}$ on the first subsystem of the $j$-th copy of $\ME$, then, no matter which outcome is obtained, measure the post-measurement state in the basis $\{\ket{v(\sigma_{\vx,\vz})}\}_{\vx,\vz \in \Zdn}$. Let $(\vx_j,\vz_j)$ be the outcome.
\item Accept $M$ if and only if $|\bigcup_{j=1}^L \supp(\vx_j,\vz_j)| \le k$. 
\een 
\label{alg:klocal}  
\end{algorithm}

For correctness, it is sufficient to show:\\
(1) if $M$ is $k$-local, then this algorithm accepts $M$ with certainty;\\
(2) if $M$ is accepted by this algorithm with probability at least $1/3$, then $M$ is $\epsilon$-close to some $k$-local measurement.\\

Before prove the two statements, observe that when we perform $M=\{M_i\}_{i \in \N}$ on the first subsystem of $\ME$, with probability $p(M_i)$ we obtain outcome $i$ and the corresponding post-measurement state is $\ket{\tilde{v}(M_i)}$. If this post-measurement state is measured in the basis $\{\ket{v(\sigma_{\vx,\vz})}\}_{\vx,\vz \in \Zdn}$, then the outcome $(\vx,\vz)$ occurs with probability $q_{\vx,\vz}(M_i)$. So, in step 3, the probability of $(\vx_j,\vz_j)$ being some $(\vx,\vz)$ is 
\beba
\xi_{\vx,\vz}(M) &\coloneqq & \sum\limits_{i \in \N} p(M_i)q_{\vx,\vz}(M_i) \\
&=& \sum\limits_{i \in \N} |\mu_{\vx,\vz}(M_i)|^2.
\label{eq:xivxvz}
\eeea
Hence, the probability of $(\vx_j,\vz_j)$ being from $\Gamma_T$
is
\beba
\eta_T(M)& \coloneqq & 
\sum\limits_{(\vx,\vz) \in \Gamma_T} \xi_{\vx,\vz}(M)\\
&=& \sum\limits_{i \in \N}
\sum\limits_{(\vx,\vz) \in \Gamma_T}
|\mu_{\vx,\vz}(M_i)|^2 \\
&=& \dfrac{1}{D}\sum\limits_{i \in \N}
\|f_T(M_i)\|_F^2.
\label{eq:pT}
\eeea
This holds for any $j=1,2,\dots,L$.

Now let us prove statement (1). Suppose $M$ is $k$-local, i.e. $M_i$'s satisfy Eq.(\ref{eq:klocaldef}) for some $T \subseteq [n]$ with $|T|=k$. Then for any $i \in \N$, $\mu_{\vx,\vz}(M_i)$ is nonzero only for $(\vx,\vz) \in \Gamma_T$. This implies that in step 3 we have $(\vx_j,\vz_j) \in \Gamma_T$ for any $j$. So $\bigcup_{j=1}^L \supp(\vx_j,\vz_j)\subseteq T$ and $M$ is always accepted by algorithm \ref{alg:klocal}.

Next, we prove statement (2). Suppose $M$ is accepted by the algorithm with probability at least $1/3$. Then we claim that there exists $T \subseteq [n]$ with $|T|=k$ such that 
\bea
\eta_T(M) \ge 1-\epsilon^2.
\label{eq:etat}
\eea 
Suppose on the contrary that for any such $T$,
$\eta_T(M)<1-\epsilon^2$. Then $M$ must be rejected by the algorithm with probability at least $0.9$. The reason is as follows. The algorithm accepts $M$ only if 
$(\vx_1,\vz_1),(\vx_2,\vz_2),\dots,(\vx_L,\vz_L)$
satisfy $|\bigcup_{j=1}^L \supp(\vx_j,\vz_j)| \le k$.
For any such $(\vx_j,\vz_j)$'s, we can find $k$ numbers $1 \le j_1<j_2<\dots<j_k \le L$ such that 
\beba
\bigcup\limits_{j=1}^L \supp(\vx_j,\vz_j)
=\bigcup\limits_{i=1}^k \supp(\vx_{j_i},\vz_{j_i}).
\eeea
(For each $a \in \bigcup_{j=1}^L \supp(\vx_j,\vz_j)$, we choose a
$(\vx_{j_i},\vz_{j_i})$ such that $\supp(\vx_{j_i},\vz_{j_i})$ contains $a$. Since there are at most $k$ elements in $\bigcup_{j=1}^L \supp(\vx_j,\vz_j)$, $k$ such $(\vx_{j_i},\vz_{j_i})$'s are sufficient.)
Now there can be $\binom{L}{k}$ possibilities for $j_1,j_2,\dots,j_k$, and 
when the $(\vx_{j_i},\vz_{j_i})$'s are determined, the total support $W\coloneqq\bigcup_{j=1}^L \supp(\vx_j,\vz_j)$ is fixed, and then any $(\vx_{j},\vz_{j})$
for $j\neq j_1,\dots,j_L$ must be chosen from $\Gamma_W$, with probability is at most $1-\epsilon^2$ by assumption. Thus, the total probability of getting a valid sequence of $(\vx_j,\vz_j)$'s is at most
$\binom{L}{k}(1-\epsilon^2)^{L-k}\le
L^k e^{-\epsilon^2(L-k)} \le 0.1$ by our choice
of $L$.

Now, by Eqs.(\ref{eq:pT}) and (\ref{eq:etat}), 
\bea
\sum\limits_{i\in \N} \|f_T(M_i)\|_F^2
\ge D(1-\epsilon^2) 
\eea
for $|T|=k$. Then, by lemma \ref{lem:klocal}, $M$ is $\epsilon$-close to a $k$-local measurement.

Finally, algorithm \ref{alg:klocal} obviously has query complexity $O(k \log (k/\epsilon)/ \epsilon^2)$ as claimed. Moreover, besides querying the black box, it requires $O(k \log (k/\epsilon)/ \epsilon^2)$ quantum operations including: 1. preparing the state $\ME$, which is equivalent to preparing the $j$-th and $(n+j)$-th qudits in the $d$-dimensional maximally entangled state
for $j=1,2,\dots,n$; 2. measuring a $2n$-qudit state in the basis $\{\ket{v(\sigma_{\vx,\vz})}\}_{\vx,\vz \in \Zdn}$, which is equivalent to measuring the $j$-th and $(n+j)$-th qubits in the basis 
$\{\ket{v(\sigma_{x,z})}\}_{x,z\in \Zd}$ for $j=1,2,\dots,n$. In addition, the classical processing is also efficient. Thus, this algorithm can be efficiently implemented.
\end{proof} 

\section{Testing the Permutation-Invariant Measurements}
\label{sec:perminv}
In this section, we consider testing a class of multi-qudit measurements possessing certain symmetry with respect to the permutations of the qudits. This kind of measurements have been widely used in many tasks such as quantum data compression  \cite{Sch95,JS94,JHH+98,HM02a,BHL06}, state estimation  \cite{VLP+99,KW01,BBG+06} and entanglement concentration  \cite{BBP+96,HM02b}. Usually we perform such measurements on $\rho^{\otimes n}$ to extract certain information from it, or to transform it into a better form, where $\rho$ is some mixed state that we are interested in. For example, in quantum data compression, suppose an i.i.d. quantum source produces a $d$-dimensional mixed state $\rho$ each time. We first collect $n$ copies of this state, and then compress $\rho^{\otimes n}$ into a smaller Hilbert space whose dimension is approximately $2^{nS(\rho)}$. Here $S(\rho)$ is the von Neumann entropy of $\rho$, i.e. if $\rho$ has the spectral decomposition 
\bea
\rho=\sum\limits_{x=0}^{d-1} p(x)\ketbra{x}{x},
\eea
then 
\bea
S(\rho)=-\sum\limits_{x=0}^{d-1}p(x)\log_2 p(x).
\eea
The typical method \cite{Sch95,JS94} is that we perform a projective measurement $\{\Pi_{\epsilon},I-\Pi_{\epsilon}\}$ on $\rho^{\otimes n}$, and if the outcome corresponds to $\Pi_{\epsilon}$, then we succeed. Here $\Pi_{\epsilon}$ is the projection operator onto the so-called ``$\epsilon$-typical subspace". This subspace is spanned by all $\ket{x_1,x_2,\dots,x_n}$'s  satisfying
\bea
\left |\dfrac{1}{n}\log_2\left(\dfrac{1}{p(x_1)p(x_2)\dots p(x_n)}\right) -S(\rho) \right| \le \epsilon.
\eea
Note that this condition depends only on how many $x_i$'s are $0$, $1$, $\dots$, $d-1$ respectively. Therefore, this $\epsilon$-typical subspace is invariant under any permutation of the $n$ qudits. In other words, the measurement $\{\Pi_{\epsilon},I-\Pi_{\epsilon}\}$ treats every qudit equally. This property is shared by many other measurements, including those used in other quantum data compression methods  \cite{JHH+98,HM02a,BHL06}, state estimation  \cite{VLP+99,KW01,BBG+06} and entanglement concentration  \cite{BBP+96,HM02b}. Thus, as a preliminary step of testing these measurements, we can test whether an unknown measurement is permutation-invariant in the first place.

Formally, let $S_n$ be the symmetric group on $n$ elements. For any $\tau \in S_n$, it can be viewed as a unitary operation on $n$ qudits as follows:
\bea
\tau \ket{\phi_1,\phi_2,\dots,\phi_n}
=\ket{\phi_{\tau^{-1}(1)}, \phi_{\tau^{-1}(2)}, \dots, \phi_{\tau^{-1}(n)}},
\eea
where the $\ket{\phi_i}$'s are arbitrary qudit states. We say an operator $A \in \BHdn$ is \textit{permutation-invariant} if
\bea
A=\tau A \tau^{\dagger},~~~\forall \tau \in S_n.
\label{eq:perm.inv.1}
\eea
Note that this condition is equivalent to
\bea
A=\sum\limits_{k} \mu_k E_k^{\otimes n}
\label{eq:perm.inv.2}
\eea
for some $\mu_k \in \C$ and $E_k \in \BHd$. A measurement $M=\{M_i\}_{i \in \N}$ on $n$ qudits is said to be permutation-invariant if $M_i$ is permutation-invariant for each $i \in \N$. In this case, if we first apply a permutation to the qudits, then perform $M$, and at last apply the inverse permutation, then the effect is equivalent to directly performing $M$.

\subsection{Representation Theory Background}
Our algorithm for testing the permutation-invariant measurements depends crucially on the following results from the representation theory of symmetric group. Let $d,n$ be arbitrary integers. Let 
\beba
\I_{d,n}&=&\{\vla=(\lambda_1,\lambda_2,\dots,\lambda_d): \lambda_i \in \Z, ~\lambda_1 \ge \lambda_2 \ge \dots  \\
&& \ge \lambda_d \ge 0,~~\sum\limits_{i=1}^d \lambda_i=n.\}
\eeea
be the set of partitions of $n$ into at most $d$ parts. A Young diagram of shape $\vla \in \I_{d,n}$ (denoted by $Y(\vla)$) is a finite collection of boxes arranged in left-justified rows, with the $i$-th row having $\lambda_i$ boxes for $i=1,2,\dots,d$. Fox example, Fig.\ref{fig:youngdiagram} illustrates the Young diagram of shape $(5,3,1)$.
\begin{figure}[H]
\centering
\setlength{\unitlength}{0.3mm}
\begin{picture}(200,80)
\put(40,0){\line(0,1){60}}
\put(60,0){\line(0,1){60}}
\put(80,20){\line(0,1){40}}
\put(100,20){\line(0,1){40}}
\put(120,40){\line(0,1){20}}
\put(140,40){\line(0,1){20}}
\put(40,60){\line(1,0){100}}
\put(40,40){\line(1,0){100}}
\put(40,20){\line(1,0){60}}
\put(40,0){\line(1,0){20}}
\end{picture}
\caption{The Young diagram of shape $(5,3,1)$.}
\label{fig:youngdiagram}
\end{figure}
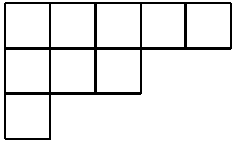
Each $\vla \in \I_{d,n}$ (or the Young diagram $Y(\vla)$) corresponds to an irreducible representation $V_{\vla}$ of the symmetric group $S_n$, and also corresponds to an irreducible representation $W_{\vla,d}$ of the general linear group $GL(d)$ of $d \times d$ invertible matrices. Furthermore, we can calculate the dimension of $V_{\vla}$ and $W_{\vla,d}$ as follows. For a box in the $i$-th row and $j$-th column of a Young diagram $Y(\vla)$, we define its \textit{hook length} $\hook(i,j)$ to be the number of boxes that are in the same row to the right of it or in the same column below it plus one. For example, Fig. \ref{fig:hooklength} illustrates the hook lengths of each box in the Young diagram of shape $(5,3,1)$.
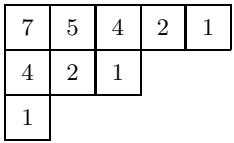
\begin{figure}[H]
\centering
\setlength{\unitlength}{0.3mm}
\begin{picture}(200,80)
\put(40,0){\line(0,1){60}}
\put(60,0){\line(0,1){60}}
\put(80,20){\line(0,1){40}}
\put(100,20){\line(0,1){40}}
\put(120,40){\line(0,1){20}}
\put(140,40){\line(0,1){20}}
\put(40,60){\line(1,0){100}}
\put(40,40){\line(1,0){100}}
\put(40,20){\line(1,0){60}}
\put(40,0){\line(1,0){20}}
\put(40,0){\makebox(20,20){1}}
\put(40,20){\makebox(20,20){4}}
\put(40,40){\makebox(20,20){7}}
\put(60,20){\makebox(20,20){2}}
\put(60,40){\makebox(20,20){5}}
\put(80,20){\makebox(20,20){1}}
\put(80,40){\makebox(20,20){4}}
\put(100,40){\makebox(20,20){2}}
\put(120,40){\makebox(20,20){1}}
\end{picture}
\caption{The hook lengths of each box in the Young diagram
of shape $(5,3,1)$}
\label{fig:hooklength}
\end{figure}

Then the dimension of $V_{\vla}$ is given by
\bea
v_{\vla} \coloneqq \dfrac{n!}{\prod\limits_{(i,j) \in Y(\vla)} \hook(i,j)},
\eea
while the dimension of $W_{\vla,d}$ is given by
\beba
w_{\vla,d}& \coloneqq &\prod\limits_{(i,j) \in Y(\vla)} \dfrac{d+j-i}{ \hook(i,j)}.
\eeea
For more details about $V_{\vla}$ and $W_{\vla,d}$, please see Refs.\cite{FH91,Ful97}.

The Schur-Weyl duality \cite{Sch01,Sch27,Wey39} states that there exists an orthonormal basis for $\Hdn$, which is called the \textit{Schur basis} and lablled by $\ket{\vla}\ket{a_{\vla}}\ket{b_{\vla}}_{\sch}$ for $\vla \in \I_{d,n}$, $0 \le a_{\vla} \le w_{\vla,d}-1$, $0 \le b_{\vla} \le v_{\vla}-1$, such that for any $\tau \in S_n$ and $E \in GL(d)$ we have
\bea
\tau \ket{\vla}\ket{a_{\vla}}\ket{b_{\vla}}_{\sch}
&=& 
\ket{\vla}\ket{a_{\vla}}(V_{\vla}(\tau)\ket{b_{\vla}})_{\sch},
\label{eq:swdual1}\\
E^{\otimes n} \ket{\vla}\ket{a_{\vla}}\ket{b_{\vla}}_{\sch}
&=& 
\ket{\vla}(W_{\vla,d}(E)\ket{a_{\vla}})\ket{b_{\vla}}_{\sch}.
\label{eq:swdual2}
\eea
In other words, the space $\Hdn$ can be decomposed as
\bea
\Hdn=\bigoplus\limits_{\vla \in \I_{d,n}}
(\W_{\vla,d} \otimes \V_{\vla})_{\sch},
\eea 
where $\W_{\vla,d}$ is $w_{\vla,d}$-dimensional and $\V_{\vla}$ is $v_{\vla}$-dimensional, such that
any $\tau \in S_n$ and $E \in GL(d)$ satisfy
\bea
\tau &=&\bigoplus\limits_{\vla \in \I_{d,n}} (I \otimes V_{\vla}(\tau))_{\sch}, 
\label{eq:dual1}\\
E^{\otimes n}
&=&\bigoplus\limits_{\vla \in \I_{d,n}} (W_{\vla,d}(E) \otimes I)_{\sch}.
\label{eq:dual2}
\eea
Since $W_{\vla,d}$ is a continuous mapping from $GL(d)$ to $GL(w_{\vla,d})$, and $GL(d)$ is a dense subset of the $d \times d$ matrices, we can naturally extend $W_{\vla,d}$ to all $d \times d$ matrices. Then for any permutation-invariant $A \in \BHdn$, if
\bea
A=\sum_k \mu_k E_k^{\otimes k}, 
\eea
let
\bea
\hat{A}_{\vla}\coloneqq
\sum_k \mu_k W_{\vla,d}(E_k).
\eea
Then
\bea
A\ket{\vla}\ket{a_{\vla}}\ket{b_{\vla}}_{\sch}
&=& 
\ket{\vla}(\hat{A}_{\vla}\ket{a_{\vla}})\ket{b_{\vla}}_{\sch},
\eea
and
\beba
A&=&\bigoplus\limits_{\vla \in \I_{d,n}} (\hat{A}_{\vla} \otimes I)_{\sch}.
\label{eq:blockdiag}
\eeea
On the other hand, it is easy to see that any $A \in \BHdn$ satisfying Eq.(\ref{eq:blockdiag}) commutes with any $\tau \in S_n$ (which satisfies Eq.(\ref{eq:dual1})), since they act on disjoint subsystems. Hence, $A \in \BHdn$ is permutation-invariant if and only if it has the block-diagonal form in Eq.(\ref{eq:blockdiag}) with respect to the Schur basis.

Finally, let $\Lambda_{d,n} \coloneqq \{(\vla,a_{\vla},b_{\vla}):\vla \in \I_{d,n},0 \le a_{\vla} \le w_{\vla,d}-1, 0 \le b_{\vla} \le v_{\vla}-1\}$, and fix a bijection $h:\Zdn \to \Lambda_{d,n}$.
Then the standard basis state $\ket{i_1,i_2,\dots,i_n}$ can be also labelled as $\ket{\vla,a_{\vla},b_{\vla}}$,  provided that $h(i_1,i_2,\dots,i_n)=(\vla,a_{\vla},b_{\vla})$. We use $U_{\sch}$ to denote the unitary operation that converts the standard basis to the Schur basis, and call it the \textit{Schur transform}  \cite{BCH07}. Namely,
if $h(i_1,i_2,\dots,i_n)=(\vla,a_{\vla},b_{\vla})$, then
\bea
U_{\sch}\ket{i_1,i_2,\dots,i_n}= 
\ket{\vla,a_{\vla},b_{\vla}}_{\sch}.
\eea
Its inverse $U_{\sch}^{-1}$ is called \textit{the inverse Schur transform}. Both $U_{\sch}$ and $U_{\sch}^{-1}$ can be (approximately) implemented with time complexity polynomial in $n$, $d$ and $\log(1/\epsilon)$, where $\epsilon$ is the accuracy parameter \cite{BCH07}. 

\subsection{Our Result on Testing the Permutation-Invariant Measurements}
It will be convenient to establish the following notations. For any $A \in \BHdn$, it can be written as 
\beba
A=\sum\limits_{\vla,\vla' \in \I_{d,n}} (\ketbra{\vla}{\vla'}
\otimes A_{\vla,\vla'})_{\sch}
\eeea
for some $A_{\vla,\vla'}$'s. We are particularly interested in the diagonal term $A_{\vla,\vla} \in \mathcal{B}(\W_{\vla,d} \otimes \V_{\vla})$.
Let $\{g_{\vla,j}\}_{0 \le j \le v_{\vla}^2-1}$ be the generalized Pauli operators on $\V_{\vla}$ such that $g_{\vla,0}=I$. Then we can expand $A_{\vla,\vla}$
as 
\beba
A_{\vla,\vla}
=\hat{A}_{\vla}\otimes I
+\sum\limits_{j \neq 0} \tilde{A}_{\vla,j} \otimes g_{\vla,j}
\eeea
for some $\hat{A}_{\vla}$ and $\tilde{A}_{\vla,j}$'s. 
So $A$ can be written as the sum of three terms:
\bea
A = (\hat{A} + \bar{A} + \tilde{A})_{\sch},
\eea
where
\bea
\hat{A} &\coloneqq &\bigoplus\limits_{\vla \in \I_{d,n}} \hat{A}_{\vla} \otimes I, \\
\tilde{A} &\coloneqq & 
\bigoplus\limits_{\vla \in \I_{d,n}}
(\sum\limits_{j \neq 0} \tilde{A}_{\vla,j} \otimes g_{\vla,j}),\\
\bar{A} &\coloneqq &\sum\limits_{\vla \neq \vla'}
\ketbra{\vla}{\vla'} \otimes A_{\vla,\vla'}.
\eea
While $\hat{A}$ is permutation-invariant, $\tilde{A}$ and $\bar{A}$ 
are not. Furthermore, $\hat{A}$, $\tilde{A}$ and $\bar{A}$ are mutually orthogonal with respect to the Hilbert-Schmidt product. Thus,
\bea
\|A\|_F^2=\|\hat{A}\|_F^2+\|\tilde{A}\|_F^2+\|\bar{A}\|_F^2.
\eea

With these notations, a measurement $M=\{M_i\}_{i \in \N}$ is permutation-invariant if and only if $M_i=(\hat{M}_i)_{\sch}$ for any $i \in \N$. The next lemma shows that if this is approximately true, then $M$ is close to a permutation-invariant measurement.

\begin{lemma}
For any $0<\delta<1$, if a measurement $M=\{M_i\}_{i \in \N}$ on $n$ qudits satisfies
\bea
\sum\limits_{i \in \N} \|\hat{M}_i\|_F^2 \ge D(1-\delta^2),
\eea
where $D=d^n$, then $M$ is $\delta$-close to a permutation-invariant measurement.
\label{lem:perminv}
\end{lemma}
\begin{proof}
See Appendix \ref{apd:lemperminv}.
\end{proof}

With the help of lemma \ref{lem:perminv}, we obtain the following result on testing the permutation-invariant measurements.

\begin{theorem}
The permutation-invariant measurements on $n$ qudits can be $\epsilon$-tested with query complexity $O(1/\epsilon^2)$. Furthermore, the testing algorithm can be efficiently implemented.
\end{theorem}
\begin{proof}
Given a black box performing some unknown measurement $M=\{M_i\}_{i \in \N}$ on $n$ qudits, we run the following test on it: 

\begin{algorithm}[H]
\caption{Testing the permutation-invariant measurements}
\ben
\item Let $D=d^n$, $L=\left\lceil \dfrac{5}{\epsilon^2}\right\rceil$.
\item Repeat the following test $L$ times:
\ben
\item Prepare 
$$~~~~~~~~~\ME=
\dfrac{1}{\sqrt{D}}\sum\limits_{(i_1,\dots,i_n)\in \Zdn}\ket{i_1,i_2,\dots,i_n}\ket{i_1,i_2,\dots,i_n}.$$
\item Perform the Schur transform on the first subsystem of $\ME$; \item Perform $M$ on the first subsystem.
\item No matter which measurement outcome occurs in step (c), perform the inverse Schur transform on the first subsystem.
\item Re-label the standard basis of both subsystems as  $\ket{\vla,a_{\vla},b_{\vla}}$'s via the bijection $h:\Zdn \to \Lambda_{d,n}$.
\item Measure the two $\ket{\vla}$ registers. If the two outcomes are not equal, then reject $M$ and quit.
\item Measure the two $\ket{b_{\vla}}$ registers in the basis $\{\ket{v(g_{\vla,j})}\}_{0 \le j \le v_{\vla}^2-1}$. If the outcome corresponds to $j \neq 0$, then reject $M$ and quit.
\een
\item Now $M$ has passed all the above tests. Accept $M$.  
\een
\label{alg:perminv}
\end{algorithm}

For correctness, we claim that this algorithm accepts $M$ with probability $(\frac{1}{D}\sum_{i \in \N}\|\hat{M_i}\|_F^2)^L$. Assuming this is true, then:\\
(1) if $M$ is permutation-invariant, then $M_i=(\hat{M_i})_{\sch}$ for any $i \in \N$. Then by $\sum_{i \in \N} M_i^{\dagger}M_i=I$, we obtain
\bea
\sum\limits_{i \in \N} \|\hat{M}_i\|_F^2=
\sum\limits_{i \in \N} \|{M}_i\|_F^2=D,
\eea
which implies that algorithm \ref{alg:perminv} accepts $M$ with certainty; \\
(2) on the other hand, if $M$ is accepted by algorithm \ref{alg:perminv} with probability at least $1/3$, then by our choice of $L$ we must have 
\bea
\sum\limits_{i \in \N}\|\hat{M_i}\|_F^2 \ge D(1-\epsilon^2)
\eea 
(because if otherwise, then the probability of $M$ being 
accepted is at most $(1-\epsilon^2)^L \le e^{-\epsilon^2L} <0.1$, contradicting our assumption). Then, by lemma \ref{lem:perminv}, $M$ is $\epsilon$-close to a permutation-invariant measurement.

Now we prove the above claim by showing that $M$ passes each iteration of step 2 with probability 
$\frac{1}{D}\sum_{i \in \N}\|\hat{M_i}\|_F^2$. 
Assuming the outcome at step 2.(c) is $i$, the unnormalized state after step 2.(e) is given by
\beba
&&[(U_{\sch}^{\dagger} M_i U_{\sch}) \otimes I] \ME  \\
&=& [(\hat{M_i}+\tilde{M_i}+\bar{M_i}) \otimes I] \ME  \\
&=&\dfrac{1}{\sqrt{D}}
\sum\limits_{\vla,a_{\vla},b_{\vla}}
((\hat{M_i}+\tilde{M_i}+\bar{M_i}) 
\ket{\vla,a_{\vla},b_{\vla}})
\ket{\vla,a_{\vla},b_{\vla}} ,
\eeea
where
\beba
\hat{M_i}\ket{\vla}\ket{a_{\vla}}\ket{b_{\vla}}
&=&\ket{\vla}((\hat{M_i})_{\vla}\ket{a_{\vla}})\ket{b_{\vla}},
\eeea
\beba
\tilde{M_i}\ket{\vla}\ket{a_{\vla}}\ket{b_{\vla}}
&=&\sum\limits_{j \neq 0}
\ket{\vla}((\tilde{M_i})_{\vla,j}\ket{a_{\vla}}) (g_{\vla,j}\ket{b_{\vla}}) ,
\eeea
\beba
\bar{M_i}\ket{\vla}\ket{a_{\vla}}\ket{b_{\vla}}
&=&\sum\limits_{\vla' \neq \vla}
\ket{\vla'} ((\bar{M_i})_{\vla',\vla}
\ket{a_{\vla}}\ket{b_{\vla}}).
\eeea
Note that $\bar{M_i}$ changes $\ket{\vla}$ and
$\tilde{M_i}$ changes $\ket{b_{\vla}}$. Only
$\hat{M_i}$ leaves both $\ket{\vla}$ and $\ket{b_{\vla}}$ intact. 
In step 2.(f), we do not reject $M$ only if the measurements on the two $\ket{\vla}$ registers yield the same outcome. In this case, assuming the outcome is $\vla$, then the unnormalized state of the rest four registers becomes
\beba
&&\dfrac{1}{\sqrt{D}}
\sum\limits_{a_{\vla},b_{\vla}}
[((\hat{M_i})_{\vla} \ket{a_{\vla}})\ket{b_{\vla}}
\ket{a_{\vla}}\ket{b_{\vla}} \\
&&+
\sum\limits_{j \neq 0}
((\tilde{M_i})_{\vla,j}\ket{a_{\vla}}) (g_{\vla,j}\ket{b_{\vla}})
\ket{a_{\vla}}\ket{b_{\vla}}] \\
&=&\dfrac{1}{\sqrt{D}}
\sum\limits_{a_{\vla}}
[((\hat{M_i})_{\vla} \ket{a_{\vla}}) \ket{a_{\vla}}
(\sum\limits_{b_{\vla}} \ket{b_{\vla}}\ket{b_{\vla}})\\
&&+
\sum\limits_{j \neq 0}
((\tilde{M_i})_{\vla,j}\ket{a_{\vla}})
\ket{a_{\vla}} 
(\sum\limits_{b_{\vla}} (g_{\vla,j}\ket{b_{\vla}})\ket{b_{\vla}})]\\
&=&
\sqrt{\dfrac{v_{\vla}}{D}}
\sum\limits_{a_{\vla}}
[((\hat{M_i})_{\vla} \ket{a_{\vla}})\ket{a_{\vla}} \ket{v(g_{\vla,0})} \\
&&+
\sum\limits_{j \neq 0}
((\tilde{M_i})_{\vla,j}\ket{a_{\vla}})
\ket{a_{\vla}} \ket{v(g_{\vla,j})}],
\eeea
where in the first step we have switched the order of the first $\ket{b_{\vla}}$ and the second $\ket{a_{\vla}}$. In this case,  the measurement in step 2.(g) produces the outcome $j=0$ with probability 
\beba
\dfrac{v_{\vla}}{D}
\|\sum\limits_{a_{\vla}}
((\hat{M_i})_{\vla} \ket{a_{\vla}})\ket{a_{\vla}}\|^2
&=&
\dfrac{v_{\vla}}{D} \|(\hat{M_i})_{\vla}\|_F^2.
\eeea
Summing this probability over all $i$'s and $\vla$'s, we know that $M$ passes one iteration of step 2 with probability
\bea
\sum\limits_{i \in \N}\sum\limits_{\vla} \dfrac{v_{\vla}}{D}\|(\hat{M_i})_{\vla}\|_F^2
=\dfrac{1}{D} \sum\limits_{i \in \N} \|\hat{M_i}\|_F^2, 
\eea
as claimed.

Finally, this algorithm obviously has query complexity $O(1/\epsilon^2)$. Furthermore, besides querying the black box, it needs $\poly(n,1/\epsilon)$ number of quantum operations including: (1)preparing the state $\ME$, which is equivalent to preparing the $j$-th and $(n+j)$-th qudits in the $d$-dimensional maximally entangled state;
(2)quantum Schur transform and its inverse, which can be efficiently implemented;
(3)measuring the two $\ket{b_{\vla}}$ registers in the basis $\{\ket{v(g_{\vla,j})}\}_{0 \le j \le v_{\vla}^2-1}$, which can be accomplished efficiently with the circuit in Fig.\ref{fig:measureb}. 
\begin{figure}[H]
\begin{equation*}
\Qcircuit @C=2em @R=2em {
& \qw & \multigate{1}{CX_{v_{\vla}}} &  \gate{F_{v_{\vla}}} & \meter \\
& \qw &  \ghost{CX_{v_{\vla}}}   & \qw                & \meter
}
\end{equation*}
\caption{The circuit for measuring the two $\ket{b_{\vla}}$ registers in the basis 
$\{\ket{v(g_{\vla,j})}\}_{0 \le j \le v_{\vla}^2-1}$ where 
the $g_{\vla,j}$'s are the $v_{\vla}$-dimensional Pauli operators. Here $CX_{v_{\vla}}$ is the $v_{\vla}$-dimensional generalization of CNOT: $CX_{v_{\vla}}\ket{a}\ket{b}=\ket{a}\ket{b+a(\textrm{mod}~v_{\vla})}$, for $a,b=0,1,\dots,v_{\vla}-1$, and $F_{v_{\vla}}$ is the $v_{\vla}$-dimensional quantum Fourier transform. Since $v_{\vla} \le d^n$, both $CX_{v_{\vla}}$ and $F_{v_{\vla}}$ can be efficiently implemented. 
}
\label{fig:measureb}
\end{figure}
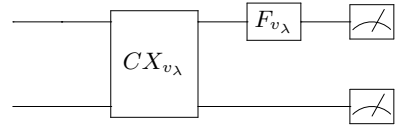
In addition, the classical processing is also easy. So algorithm \ref{alg:perminv} can be efficiently implemented.
\end{proof}

\section{Testing Any Finite Set of Measurements}
\label{sec:finite}
So far we have studied testing three particular classes of quantum measurements. In this section, we present a general result about property testing of quantum measurements. Specifically, we give an algorithm that can test any finite set of measurements on any finite-dimensional system.

Before stating our result, it is helpful to consider the following question. Suppose $M=\{M_i\}_{i \in \N}$ and $N=\{N_i\}_{i \in \N}$ are two measurements on a $D$-dimensional system. If a black box performs one of them, how do we know which one it really implements? We solve this problem by repeating the following experiment sufficiently many times: each time we prepare a copy of $\ME$ and apply the unknown measurement to its first subsystem, obtaining an outcome and post-measurement state. If the unknown measurement is $M$ (or $N$), then the outcome $i$ occurs with probability $p(M_i)$
(or $p(N_i)$) and the corresponding post-measurement state would be $\ket{\tilde{v}(M_i)}$ (or $\ket{\tilde{v}(N_i)}$). We claim that as long as $\Delta(M,N)$ is large, we do not need to repeat this experiment many times to successfully identify the unknown measurement. The reason is as follows.

Let $\vp(M) \coloneqq (p(M_i))_{i \in \N}$ and $\vp(N) \coloneqq (p(N_i))_{i \in \N}$ be the distributions of outcomes for $M$ and $N$ respectively. Consider their variational distance
\bea
\Delta(\vp(M),\vp(N))
=\dfrac{1}{2}\sum\limits_{i \in \N} |p(M_i)-p(N_i)|.
\eea
By Eqs.(\ref{eq:DeltaMN2}), (\ref{eq:st&f}), (\ref{eq:isoip}) and the Cauchy-Schwarz inequality, 
\beba
\Delta^2(M,N) &=& 1- \sum\limits_{i \in \N} |\bracket{v(M_i)}{v(N_i)}| \\
&\ge &
1- \sum\limits_{i \in \N} \|\ket{v(M_i)}\|\cdot \|\ket{v(N_i)}\| \\
&= &
1-\sum\limits_{i \in \N} \sqrt{p(M_i)p(N_i)} \\
&=&
1-F(\vp(M),\vp(N)) \\
&\ge & \dfrac{1}{2}D^2(\vp(M),\vp(N)).
\eeea
Thus
\bea
\Delta(M,N) \ge \dfrac{1}{\sqrt{2}}D(\vp(M),\vp(N)).
\label{eq:deltadelta}
\eea
This means that if $D(\vp(M),\vp(N))$ is large, then $\Delta(M,N)$ is also large. In this case, we can simply distinguish $M$ and $N$ from the outcome statistics.

But what if $\Delta(M,N)$ is large but $D(\vp(M),\vp(N))$ is small? In this case, we need to use the post-measurement states. Specifically, suppose $|M|,|N|\le k$. We repeat the above experiment $L$ times, and suppose the outcome $i$ occurs $L_i$ times for $i=1,2,\dots,k$. Note that with high probability we have $L_i \approx p(M_i)L \approx p(N_i)L$ for each $i$. Then let
$\vL=(L_1,L_2,\dots,L_k)$ and
\bea
\chi_{\vL}(M) \coloneqq \bigotimes\limits_{i=1}^k (\ket{\tilde{v}(M_i)}^{\otimes L_i}),\\
\chi_{\vL}(N) \coloneqq \bigotimes\limits_{i=1}^k (\ket{\tilde{v}(N_i)}^{\otimes L_i}).
\eea
The following lemma says that if $\Delta(M,N)$ is large, then with high probability $\chi_{\vL}(M)$ and $\chi_{\vL}(N)$ have small overlap for some small $L$, and thus can be easily distinguished.

\begin{lemma}
Let $\vL=(L_1,L_2,\dots,L_k)$ be a set of non-negative integers with $\sum_{i=1}^k L_i=L$. Suppose $M=\{M_1,M_2,\dots,M_k\}$ and $N=\{N_1,N_2,\dots,N_k\}$ are two measurements such that $\Delta(M,N) \ge \delta$ for some $0<\delta<1$. If for any $i$ such that $L_i \ge 0.1\delta^2L/k$, we have
\bea
p(M_i),p(N_i) \ge (1-0.1\delta^2)\dfrac{L_i}{L},
\eea
then 
\bea
|\bracket{\chi_{\vL}(M)}{\chi_{\vL}(N)}|\le (1-0.6\delta^2)^L.
\eea
\label{lem:chimchin}
\end{lemma}
\begin{proof}
See Appendix \ref{apd:lemchimchin}.
\end{proof}

Now let us return to the problem of testing any finite set of measurements on a finite-dimensional system. Suppose $\S=\{M^{(1)},M^{(2)},\dots,M^{(m)}\}$ is the set of measurements to be tested, where $M^{(i)}=\{M^{(i)}_j\}_{j \in \N}$. Given an unknown measurement $M$, we want to know whether $M$ is in $\S$ or far away from $\S$. Our basic idea is to repeat the above experiment sufficiently many times. If $M$ is from $\S$, then with high probability the outcome statistics is well-behaved and the post-measurement states lie in certain subspace. On the other hand, if $M$ is far from any $M^{(i)}$, then we divide $\S$ into two subsets: for  one subset, the outcome statistics are quite different for $M$ and any $M^{(i)}$ in it, so we can easily distinguish $M$ from this subset; for the other subset, the outcome statistics are similar for $M$ and any $M^{(i)}$ in it, so we must utilize the post-measurement states. Lemma \ref{lem:chimchin} says that $\chi_{\vL}(M)$ and $\chi_{\vL}(M^{(i)})$ are almost orthogonal for some small $L$. Then by the following lemma, $\chi_{\vL}(M)$ has a small projection onto the subspace spanned by these $\chi_{\vL}(M^{(i)})$'s. So by a projective measurement, we can get to know whether $M$ is in $\S$ or far away from $\S$.

\begin{lemma}
Suppose $\ket{\phi_1},\ket{\phi_2},\dots,\ket{\phi_m},\ket{\psi}
\in \HD$ (where $D$ is arbitrary) satisfy $|\bracket{\psi}{\phi_i}|\le 1/(5m)$ for any $i$ and $|\bracket{\phi_i}{\phi_j}|\le 1/(5m)$ for any $i \neq j$. Let $\Pi$ be the projection operator onto the subspace spanned by $\ket{\phi_1}$, $\ket{\phi_2}$, $\dots$, $\ket{\phi_m}$. Then 
$\bra{\psi}\Pi\ket{\psi} \le 0.1$.
\label{lem:pi}
\end{lemma}
\begin{proof}
See Appendix \ref{apd:lempi}.
\end{proof}

Now we formally state our result on testing any finite set of measurements on any finite-dimensional system.

\begin{theorem}
Suppose $\S=\{M^{(1)},M^{(2)},\dots,M^{(m)}\}$ is a set of 
measurements on a $D$-dimensional system, where $M^{(i)}=\{M^{(i)}_1,M^{(i)}_2,\dots,M^{(i)}_k\}$ has at most $k$ possible outcomes for any $i$, and $\Delta(M^{(i)},M^{(j)}) \ge \gamma$ for any $i \neq j$, for some $0<\gamma<1$. Then $\S$ can be $\epsilon$-tested with query complexity 
$O(\max\{k^2\log(k)/a^8,\log(m)/a^2\})$ where $a=\min\{\epsilon,\gamma\}$.
\label{thm:general}
\end{theorem}

\begin{proof}
Given a black box performing some unknown measurement $M=\{M_j\}_{j \in \N}$, we run the following test on it: 

\begin{algorithm}[H]
\caption{Testing $\S=\{M^{(1)},M^{(2)},\dots,M^{(m)}\}$}
\ben
\item Let $L=\max\left\{ \left\lceil \dfrac{5000k^2 \ln(20k)}{a^8} \right\rceil, \left \lceil \dfrac{2\ln(5m)}{a^2} \right \rceil \right \}$.
\item Prepare $L$ copies of $\ME=\dfrac{1}{\sqrt{D}}\sum\limits_{i=0}^{D-1}\ket{i}\ket{i}$. 
\item Perform $M=\{M_i\}_{i \in \N}$ on the first subsystem of each copy of $\ME$. Let $L_j$ be the number of occurrences of outcome $j$ for any $j \in \N$, and let 
$\vL=(L_1,L_2,\dots,L_k)$.
\item If $L_j>0$ for some $j>k$, then reject $M$ and quit; otherwise, define $\T \subseteq \S$ as follows: $M^{(i)} \in \T$ if and only if for any $j$ such that $L_j \ge 0.1a^2 L/k$,
we have $p(M^{(i)}_j) \ge (1-0.1a^2)L_j/L$. 
If $\T=\varnothing$, then reject $M$ and quit. 
\item Now suppose $\T=\{M^{(i_1)},M^{(i_2)},\dots,M^{(i_t)}\}$ for some $1 \le t \le m$. Let $\Pi$ be the projection operator onto the subspace spanned by $\chi_{\vL}(M^{(i_1)})$, $\chi_{\vL}(M^{(i_2)})$, $\dots$,  $\chi_{\vL}(M^{(i_t)})$. Perform the projective measurement $\{\Pi,I-\Pi\}$ on the state $\chi_{\vL}(M)$. If the outcome corresponds to $\Pi$, then accept $M$; otherwise, reject $M$.
\een 
\label{alg:general} 
\end{algorithm} 

Before prove the correctness of this algorithm, observe that, in step 3,  we should obtain outcome $j$ with probability $p(M_j)$, for any $j$. Then by Eq.(\ref{eq:chernoff2}) and our choice of $L$,  
\beba
\Pr\left[\left|\dfrac{L_j}{L}-p(M_j )\right|\le \dfrac{a^4}{100k}  \right] \ge 1-\dfrac{1}{10k}, 
\eeea
for any $j$. Thus, with probability at least $0.9$, for any $j$ satisfying $L_j/L \ge 0.1a^2/k$, we have 
\beba 
 p(M_{j}) \ge \dfrac{L_j}{L}-\dfrac{a^4}{100k}  \ge (1-0.1a^2)\dfrac{L_j}{L}
 \eeea
 
Now suppose that $M=M^{(i)}$ for some $i$. By the above observation, $M^{(i)}$ is in the set $\T$ with probability at least $0.9$. Then,  $\chi_{\vL}(M)=\chi_{\vL}(M^{(i)})$ is obviously in the subspace spanned by the $\chi_{\vL}(M^{(i_j)})$'s and hence the measurement in step 5 always produces the outcome corresponding to $\Pi$. Therefore, $M=M^{(i)}$ is accepted by algorithm \ref{alg:general} with probability at least $0.9$.

On the other hand, suppose $M$ is accepted by algorithm \ref{alg:general} with probability at least $1/3$. We need to prove that $M$ is $\epsilon$-close to some $M^{(i)}$. Suppose on the contrary that this is not true. Then, still by the above observation,  with probability at least $0.9$, for any $j$ satisfying $L_j/L \ge 0.1a^2/k$, we have $p(M_{j}) \ge (1-0.1a^2)L_j/L$. Furthermore, by the definition of $\T$, any $M^{(i_l)} \in \T$ also satisfies $p(M^{(i_l)}_{j}) \ge (1-0.1a^2)L_j/L$ for such $j$'s. 
Then by lemma \ref{lem:chimchin} and our choice of $L$,
we obtain 
\beba
|\bracket{\chi_{\vL}(M)}{\chi_{\vL}(M^{(i_l)})}| &\le & (1-0.6a^2)^L  \\
&\le & e^{-0.6a^2L}  \\
&\le & \dfrac{1}{5m},
\eeea
for any $l$, and similarly
\beba
|\bracket{\chi_{\vL}(M^{(i_l)})}{\chi_{\vL}(M^{(i_{l'})})}| &\le & \dfrac{1}{5m},  
\eeea
for any $l \neq l'$. Thus, by lemma \ref{lem:pi}, in step 5 we obtain the outcome corresponding to $\Pi$ with probability at most $0.1$. Overall, $M$ can be accepted by  algorithm \ref{alg:general} with probability at most $0.2$, contradicting our assumption.

Finally, the query complexity of algorithm \ref{alg:general} is $O(\max\{k^2\log(k)/a^8,\log(m)/a^2\})$ as claimed. In particular, if the system consists of $n$ qudits, then the query complexity is polynomial (in $n$) as long as $k$ and $1/a$ are at most polynomial and $m$ is at most exponential (in $n$). 
\end{proof}

In general, algorithm \ref{alg:general} may be not efficiently implementable. But from an information-theoretic point of view, theorem \ref{thm:general} shows that we do not need to query the black box too many times to know whether it belongs to some finite set or is far away from this set, as long as the set is not too large (i.e. $m$ is not too large) , the measurements in the set are well-separated (i.e. $\gamma$ is not too small), and any measurement in the set does not have too many possible outcomes (i.e. $k$ is not too large).

Finally, let us apply theorem \ref{thm:general} to the set of stabilizer measurements.  This set contains $4^n$ elements. The distance between any two of them is $1/2$, i.e.
$\Delta(P(\va,\vb),P(\vc,\vd))=1/2$ for any $(\va,\vb) \neq (\vc,\vd) $. Furthermore, each stabilizer measurement has only $2$ possible outcomes. Thus, by theorem \ref{thm:general}, the stabilizer measurements can be $\epsilon$-tested with query complexity $O(\max\{1/\epsilon^8,n/\epsilon^2\})$. This query complexity is only slightly worse than the $O(1/\epsilon^4)$ given by theorem \ref{thm:stabilizer}. This example shows that although algorithm \ref{alg:general} uses only certain distance information about the measurements to be tested, it still can achieve a quite good efficiency.

\section{Estimating the Distance of Two Unknown Measurements}
\label{sec:estdist}
In the previous sections, we have studied testing several properties of a single  measurement. In this section, we consider a different scenario. Suppose we are given two black-box measurement devices, how do we know whether they perform the same measurement? And if not, how do we estimate their distance? We will give an efficient algorithm for this task. Surprisingly, its query complexity does not depend on the dimension of the system, but depends only on the proximity parameter and the number of possible outcomes for the two measurements.

%basic idea
Our basic idea is as follows. Suppose $M=\{M_i\}_{1 \le i \le k}$ and $N=\{N_i\}_{1 \le i \le k}$ are two measurements with at most $k$ possible outcomes. Since
\bea
\Delta^2(M,N)=1-\sum\limits_{i=1}^k |\bracket{v(M_i)}{v(N_i)}|.
\eea
So we can estimate $\Delta^2(M,N)$ by
estimating each $|\bracket{v(M_i)}{v(N_i)}|$. Note that
\beba
|\bracket{v(M_i)}{v(N_i)}|=\sqrt{p(M_i)p(N_i)}|\bracket{\tilde{v}(M_i)}{\tilde{v}(N_i)}|.
\eeea
Thus, we can estimate $|\bracket{v(M_i)}{v(N_i)}|$ by estimating
$p(M_i)$, $p(N_i)$ and $|\bracket{\tilde{v}(M_i)}{\tilde{v}(N_i)}|$.
Recall that $p(M_i)$ (or $p(N_i)$) is the probability of obtaining outcome $i$ when we perform $M$ (or $N$) on the first subsystem of $\ME$, and $\ket{\tilde{v}(M_i)}$ (or $\ket{\tilde{v}(N_i)}$) is the corresponding post-measurement state. So we repeat this experiment many times. From the outcome statistics, we can estimate $p(M_i)$ and $p(N_i)$ with good precision. As to $|\bracket{\tilde{v}(M_i)}{\tilde{v}(N_i)}|$, we can estimate it by performing the swap test  \cite{BCW+01} on the states $\ket{\tilde{v}(M_i)}$ and $\ket{\tilde{v}(N_i)}$. In fact, the swap test can be used to estimate the overlap between any two pure states, as stated by the following lemma.

%estimate the inner product by swap test
\begin{lemma}
Given $O(\log(1/\delta)/\epsilon^4)$ copies of $\ket{\phi}$ and $\ket{\psi}$, we can estimate $|\bracket{\phi}{\psi}|$ with precision $\epsilon$ and confidence $1-\delta$. Furthermore, the estimation algorithm can be efficiently implemented.
\label{lem:swaptest}
\end{lemma}
\begin{proof}
Given any $\ket{\phi}$ and $\ket{\psi}$, the swap test is the standard procedure to estimate $|\bracket{\phi}{\psi}|$. It works as follows: first, we prepare a qubit in the state $\ket{+}=1/\sqrt{2}(\ket{0}+\ket{1})$; then, we perform a controlled-swap gate on $\ket{\phi}$ and $\ket{\psi}$, using $\ket{+}$ as the control qubit (a swap gate is the operation
$\ket{\phi}\ket{\psi} \to \ket{\psi}\ket{\phi}$); finally, we apply a Hadamard gate on the control qubit and measure it in the standard basis. The circuit for this procedure is illustrated in Fig. \ref{fig:swap}. 

\begin{figure}[H]
\begin{equation*}
\Qcircuit @C=2em @R=2em {
\lstick{\ket{0}}  & \gate{H} &  \ctrl{1} &  \gate{H} & \meter \\
\lstick{\ket{\phi}} & \qw & \multigate{1}{\textrm{Swap}}  & \qw & \qw \\
\lstick{\ket{\psi}} & \qw & \ghost{\textrm{Swap}} & \qw & \qw
}
\end{equation*}
\caption{The swap test}
\label{fig:swap}
\end{figure}
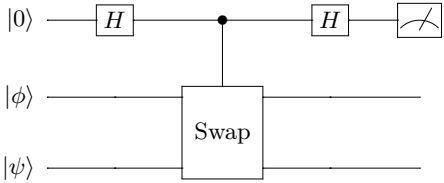
It is easy to see that the final state before the measurement is 
\bea
\dfrac{1}{2}\ket{0}(\ket{\phi}\ket{\psi}+\ket{\psi}\ket{\phi})
+\dfrac{1}{2}\ket{1}(\ket{\phi}\ket{\psi}-\ket{\psi}\ket{\phi}).
\eea
Thus, the measurement on the first qubit yields outcome $0$ with the probability $(1+|\bracket{\phi}{\psi}|^2)/2$. 

Now consider the following estimation algorithm: 

\begin{algorithm}[H]
\caption{Estimating $|\bracket{\phi}{\psi}|$}
\ben
\item Assume we have $L=\lceil 2\ln(2/\delta)/\epsilon^4 \rceil$ copies of $\ket{\phi}$ and $\ket{\psi}$. We run the swap test on each pair of $\ket{\phi}$ and $\ket{\psi}$. Let $\hat{p}_0$ be the fraction of outcome $0$ among the outcomes obtained. Then return $\sqrt{2\hat{p}_0-1}$ as the estimate of $|\bracket{\phi}{\psi}|$.
\een
\label{alg:swaptest}
\end{algorithm}

By Eq.(\ref{eq:chernoff2}) and our choice of $L$, $\hat{p}_0$ will be $\epsilon^2/2$-close to $(1+|\bracket{\phi}{\psi}|^2)/2$ with probability at least $1-\delta$. This implies that $2\hat{p}_0-1$ is $\epsilon^2$-
close to $|\bracket{\phi}{\psi}|^2$ with probability at least $1-\delta$, and hence $\sqrt{2\hat{p}_0-1}$ is $\epsilon$-close to $|\bracket{\phi}{\psi}|$ with probability at least $1-\delta$ (note that $|a-b|^2 \le |a^2-b^2|$ for any $a,b\in [0,1]$.)

The swap test requires three kinds of quantum operations including the Hadamard gate, controlled-swap gate and the measurement in the standard basis. If $\ket{\phi}$ and $\ket{\psi}$ are $n$-qudit states, then the controlled-swap gate can be implemented with $O(n)$ basic quantum gates. Algorithm \ref{alg:swaptest} repeats the swap test 
$O(\log(1/\delta)/\epsilon^4)$ times. So it can be efficiently implemented. 
\end{proof}

%ingore small measurement opertors
So, if we have sufficiently many copies of $\ket{\tilde{v}(M_i)}$ and $\ket{\tilde{v}(N_i)}$, then we can estimate their overlap with good precision. But what if $p(M_i)$ or $p(N_i)$ is so small that the outcome $i$ almost never occurs, unless we repeat the experiment extremely many times? Note that for such $i$'s, $\|M_i\|_F$
or $\|N_i\|_F$ must be small. The next lemma shows that the total contribution of such $|\bracket{v(M_i)}{v(N_i)}|$'s to $\Delta^2(M,N)$ is small, and hence we can safely ignore them.

\begin{lemma}
Suppose $M=\{M_1,M_2,\dots,M_k\}$ and $N=\{N_1,N_2,\dots,N_k\}$ are two measurements. Let $\I_{\delta}=\{i: p(M_i) \le \delta~\textrm{or}~p(N_i) \le \delta\}$ for any $0<\delta<1$ . Then 
\bea
\sum\limits_{i \in \I_{\delta}} |\bracket{v(M_i)}{v(N_i)}|
\le 2\sqrt{\delta k}.
\eea
\label{lem:smallvmivni}
\end{lemma}
\begin{proof}
Let $A=\{i: p(M_i)\le \delta\}$ and $B=\{i: p(N_i)\le \delta\}$.
Then $\I_{\delta}=A \cup B$. So it is sufficient to show
\bea
\sum\limits_{i \in A} |\bracket{v(M_i)}{v(N_i)}|
\le \sqrt{\delta k},
\label{eq:sumvmivni1}\\
\sum\limits_{i \in B} |\bracket{v(M_i)}{v(N_i)}|
\le \sqrt{\delta k},
\label{eq:sumvmivni2}.
\eea
To prove Eq.(\ref{eq:sumvmivni1}), we use the Cauchy-Schwarz inequality and the definition of $A$,
\beba
\sum\limits_{i \in A} |\bracket{v(M_i)}{v(N_i)}|
&\le &
\sum\limits_{i \in A} \|\ket{v(M_i)}\| \cdot \|\ket{v(N_i)}\| \\
&=&
\sum\limits_{i \in A} \sqrt{p(M_i)p(N_i)} \\
&\le &
\sqrt{\delta} \sum\limits_{i \in A} \sqrt{p(N_i)} \\
&\le &\sqrt{\delta k},
\eeea
where the last step follows from $\sum\limits_{i \in A}p(N_i) \le 1$ and $|A|\le k$.
Eq.(\ref{eq:sumvmivni2}) can be proved similarly.

\end{proof}

Now we are ready to formally describe our result.

\begin{theorem}
Suppose $M=\{M_1,M_2,\dots,M_k\}$ and $N=\{N_1,N_2,\dots,N_k\}$ are two measurements with at most $k$ possible outcomes on a $D$-dimensional system. Then we can estimate their distance (i.e. $\Delta(M,N)$) with precision $\epsilon$ and confidence $0.8$ by using them $O(k^5\log(k)/\epsilon^{12} )$ times. Furthermore, the estimation algorithm can be efficiently implemented provided that $k$ is polylogarithmic in $D$.
\label{thm:estdist}
\end{theorem}
\begin{proof}

We run the following algorithm on $M$ and $N$: 

\begin{algorithm}[H]
\caption{Estimating $\Delta(M,N)$} 
\ben
\item Let $L=\left\lceil \dfrac{50000 k^5 \ln(40k)}{\epsilon^{12}} \right\rceil$, $T=\left\lfloor \left(\dfrac{\epsilon^4}{16k}-\dfrac{\epsilon^4}{36k^2} \right)L \right\rfloor$.
\item Prepare $L$ copies of $\ME=\dfrac{1}{\sqrt{D}}\sum\limits_{i=0}^{D-1}\ket{i}\ket{i}$. Perform $M$ on the first subsystem of each copy. Let $a_i$ be the fraction of outcome $i$ among the $L$ outcomes obtained, for $i=1,2,\dots,k$. Let $A=\{i: a_i \ge \dfrac{\epsilon^4}{16k}-\dfrac{\epsilon^4}{36k^2}\}$.
\item 
Prepare another $L$ copies of $\ME$. Perform $N$ on the first subsystem of each copy. Let 
$b_i$ be the fraction of outcome $i$ among the $L$ outcomes obtained, for $i=1,2,\dots,k$.
Let $B=\{i: b_i \ge \dfrac{\epsilon^4}{16k}-\dfrac{\epsilon^4}{36k^2}\}$.
\item For any $i \in A \cap B$, we have at least $T$ copies of 
$\ket{\tilde{v}(M_i)}$ and $\ket{\tilde{v}(N_i)}$. Then we run algorithm \ref{alg:swaptest} on them, and get an estimate 
$\lambda_i$ of $|\bracket{\tilde{v}(M_i)}{\tilde{v}(N_i)}|$.
\item Return $\hat{\Delta}\coloneqq \sqrt{\sum\limits_{i \in A \cap B} \sqrt{a_ib_i}\lambda_i}$
as an estimate of $\Delta(M,N)$.
\een
\label{alg:estdist}
\end{algorithm}

Now we prove the correctness of this algorithm. Note that in step 2, by Eq.(\ref{eq:chernoff2}) and out choice of $L$,  $a_i$ is $ \epsilon^4/(36k^2)$-close to $p(M_i)$ with probability at least $1-1/(20k)$ for each $i$. Therefore, with probability at least $0.95$, every $i$ satisfying $p(M_i) \ge \epsilon^4/(16k)$ is in $A$. Similarly, in step 3, with probability at least $0.95$, every $i$ satisfying $p(N_i) \ge \epsilon^4/(16k)$ is in $B$.  Furthermore, by the proof of lemma \ref{lem:swaptest} and our choice of $T$, for each $i \in A \cap B$, the $\lambda_i$ produced by algorithm \ref{alg:swaptest} is $\epsilon^2/(6k)$-close to 
$|\bracket{\tilde{v}(M_i)}{\tilde{v}(N_i)}|$ with probability at least $1-1/(10k)$.
Thus, by a union bound, with probability at least $0.8$, we simultaneously have: 
\ben
\item $|a_i-p(M_i)| \le \dfrac{\epsilon^4}{36k^2}$, $\forall i$;
\item $|b_i-p(N_i)| \le \dfrac{\epsilon^4}{36k^2}$, $\forall i$;
\item $p(M_i) \le \dfrac{\epsilon^4}{16k}$, $\forall i \not\in A$;
\item $p(N_i) \le \dfrac{\epsilon^4}{16k}$, $\forall i \not\in B$;
\item $|\lambda_i-|\bracket{\tilde{v}(M_i)}{\tilde{v}(N_i)}|| \le 
\dfrac{\epsilon^2}{6k}$, $\forall i \in A\cap B.$
\een
Using these facts and lemma \ref{lem:smallvmivni}, we obtain
\beba
&&|\hat{\Delta}^2-\Delta^2(M,N)| \\
&=&
|(\sum\limits_{i \in A\cap B} |\bracket{v(M_i)}{v(N_i)}|-\sum\limits_{i \in A\cap B} \sqrt{a_ib_i}\lambda_i) \\
&&+\sum\limits_{i \not\in A\cap B} |\bracket{v(M_i)}{v(N_i)}|| \\
&\le &
\sum\limits_{i \in A\cap B} 
||\bracket{v(M_i)}{v(N_i)}|-\sqrt{a_ib_i}\lambda_i| + \dfrac{\epsilon^2}{2} \\
&\le & \dfrac{\epsilon^2}{2}+\dfrac{\epsilon^2}{2} \\
&=&\epsilon^2,
\eeea
where in the third step we have used 
\beba
&&||\bracket{v(M_i)}{v(N_i)}|-\sqrt{a_ib_i}\lambda_i| \\
&\le & |\sqrt{p(M_i)}-\sqrt{a_i}| 
+|\sqrt{p(N_i)}-\sqrt{b_i}| \\
&&+||\bracket{\tilde{v}(M_i)}{\tilde{v}(N_i)}| -\lambda_i| \\
&\le & \dfrac{\epsilon^2}{6k}+\dfrac{\epsilon^2}{6k}+\dfrac{\epsilon^2}{6k} \\
&= & \dfrac{\epsilon^2}{2k},~~~\forall i \in A \cap B.
\eeea
(Note that for any $0 \le \alpha_1,\alpha_2,
\beta_1,\beta_2,\gamma_1,\gamma_2 \le 1$, we have 
$|\alpha_1 \beta_1 \gamma_1
-\alpha_2 \beta_2 \gamma_2|
\le 
|\alpha_1
-\alpha_2|
+|\beta_1
-\beta_2|
+|\gamma_1
-\gamma_2|$).
Thus,
\beba
|\hat{\Delta}-\Delta(M,N)|  \le \epsilon
\eeea
as desired.

Finally, this algorithm obviously has query complexity $O(k^5 \log(k)/\epsilon^{12})$. Moreover, besides querying the black boxes, it requires two kinds of quantum operations: 1.preparing $O(k^5 \log(k)/\epsilon^{12})$ copies of $\ME$; 2. running algorithm \ref{alg:swaptest} $O(k)$ times. $\ME$ can be easily prepared, and algorithm \ref{alg:swaptest} is also efficient. In addition, the classical processing is also efficient. So as long as $k$ is polylogarithmic in $D$, algorithm \ref{alg:estdist} can be efficiently implemented.
\end{proof}

As a corollary of theorem \ref{thm:estdist}, we can easily determine whether two unknown measurements are the same or quite different.

\begin{corollary}
Suppose $M=\{M_1,M_2,\dots,M_k\}$ and $N=\{N_1,N_2,\dots,N_k\}$ are two measurements with $k$ possible outcomes on a $D$-dimensional system. Assuming that they are either the same or $\epsilon$-far away, we can know which case it is with probability at least $0.8$ by using them $O(k^5\log(k)/\epsilon^{12} )$ times. Furthermore, the testing algorithm can be efficiently implemented provided that $k$ is polylogarithmic in $D$.
\end{corollary}
\begin{proof}
We simply use algorithm \ref{alg:estdist} to estimate $\Delta(M,N)$ with precision $\epsilon/2$. If it is smaller than $\epsilon/2$, then we conclude that $M$ and $N$ are the same; otherwise, we conclude that $M$ and $N$ are $\epsilon$-far away. By theorem \ref{thm:estdist}, we succeed with probability at least $0.8$. In addition, the query complexity and time complexity are as claimed.
\end{proof}

\section{Conclusion}
\label{sec:conclude}

To summarize, we have introduced a metric $\Delta$ for quantum measurements on finite-dimensional systems. This metric  indicates the average difference between the behaviors of two measurements on a random input state. Then we show that, with respect to this metric, the stabilizer measurements, $k$-local measurements and permutation-invariant measurements can be all efficiently tested with query complexity independent of the system's dimension. Moreover, we also present an algorithm for testing any finite set of measurements. Finally, we give an efficient algorithm for estimating the distance of two unknown measurements, and its query complexity is also independent of the system's dimension. As a consequence, we can easily test whether two unknown measurements are identical or quite different. 

%entanglement is important
It is worth noting that entanglement plays a crucial role in all of our testing  algorithms. Namely, we need to prepare the maximally entangled state $\ME$, then ``imprint" the measurement operator $M_i$ on it, obtaining the state $\ket{v(M_i)}=(M_i \otimes I)\ME$, and finally extract information about $M_i$ from this state. It seems that by  utilizing entanglement, we can somehow gain a better  understanding of the global property of a measurement. We suspect that the three classes of measurements considered in this paper cannot be tested with the same efficiency if entanglement is not allowed.

%open questions
Finally, we would like to point out several directions that deserve further investigation:

%improve our results and lower bound
First, in this paper we have focused on giving \textit{upper bounds} for the query complexity of testing quantum measurements. It remains open to prove \textit{lower bounds} for the query complexity of the same task. In particular, it would be interesting to know whether our testing algorithms are optimal and, if they are not, how to improve them. 

%how about povm?
Second, in this paper we have assumed that we have access to both the measurement outcome and post-measurement state. For some applications, however, the post-measurement state is inaccessible and people are merely interested the outcome. In such cases, the measurement can be more conveniently described by a POVM $\{E_1,E_2,\dots,E_k\}$ where $E_i \ge 0$ and $\sum_{i=1}^k E_i=I$. It is worth studying the testing of POVMs. Note that in this new problem, we can extract information only from the outcome statistics. So it seems to be harder than the one considered in this paper, and requires quite different techniques.

%how about quantum channel?
Third, it would be also interesting to study the testing of quantum channels. Quantum measurements and quantum channels are closely related. Namely, if we perform a measurement $\{M_1,M_2,\dots,M_k\}$ on a state $\rho$ but the outcome is unknown, then this process can be equivalently characterized as a quantum channel transforming $\rho$ into $\sum_{i=1}^k M_i \rho M_i^{\dagger}$. On the contrary to testing POVMs, in testing quantum channels we can obtain information only from the mixed output states. So it also seems to be more difficult than the one considered in this paper, and needs new  techniques.

%other metric
At last, the query complexity of property testing depends crucially on the metric used. Here we have considered $\Delta$
which is a very natural and reasonable distance function for  quantum measurements. It would be interesting to study the testing of quantum measurements with respect to other  reasonable metrics.

\section*{Acknowledgments}
This research was supported by NSF Grant CCR-0905626
and ARO Grant W911NF-09-1-0440.

\begin{appendix}

\section{Proof of Lemma \ref{lem:pauliclose}}
\label{apd:pauliclose}
Without loss of generality, we can assume that 
$\mu_{\vzero,\vzero}(M_{1})$
and
$\mu_{\vzero,\vzero}(M_{2})$
are both real and non-negative (if not, we can multiply
$M_{1}$ or $M_{2}$ by appropriate phases to make this true). Then by Eqs.(\ref{eq:condm11}) and (\ref{eq:condm21}) 
we can assume
$\mu_{\vzero,\vzero}(M_{1})=\alpha_1$,
$\mu_{\vzero,\vzero}(M_{2})=\alpha_2$,
$\mu_{\va,\vb}(M_{1})=\beta_1 e^{i\theta_{1}}$
and
$\mu_{\va,\vb}(M_{2})=\beta_2 e^{i\theta_{2}}$
where 
\bea
\alpha_1,\alpha_2,\beta_1,\beta_2 \in \left[\sqrt{\dfrac{1}{4}-\gamma},\sqrt{\dfrac{1}{4}+\gamma} \right]
\label{eq:alphabeta+-}
\eea
and $\theta_{1},\theta_{2} \in [0,2\pi)$. Plugging this into Eqs.(\ref{eq:condm12}) and (\ref{eq:condm22}) we obtain
\bea
\cos(\theta_1) \ge \dfrac{1/4-\delta}{1/4+\gamma}
\ge  1-4\delta-4\gamma,\\
\label{eq:theta+}
\cos(\theta_2) \le -\dfrac{1/4-\delta}{1/4+\gamma}
\le  -1+4\delta+4\gamma.
\label{eq:theta-}
\eea
Meanwhile, combining Eqs.(\ref{eq:pau.sum.1}), (\ref{eq:condm11}) and (\ref{eq:condm21}) yields
\beba 
\sum\limits_{i=1,2} \sum\limits_{\substack{
(\vx,\vz) \neq (\vzero,\vzero),(\va,\vb)}}|\mu_{\vx,\vz}(M_i)|^2+\sum\limits_{i \ge 3} \sum\limits_{\vx,\vz \in \Ztn}|\mu_{\vx,\vz}(M_i)|^2 \\ 
\le 4\gamma. 
\label{eq:tailmu}
\eeea
Now consider the distance between $M_i$ and $P_i(\va,\vb)$ for different $i$'s:
\beba
\Delta^2(M_{1},P_{1}(\va,\vb)) 
&\le & 
\dfrac{1}{2D}\|M_{1}-P_{1}(\va,\vb)\|_F^2\\
&=&
\dfrac{1}{2}\left|\alpha_1-\dfrac{1}{2} \right|^2
+
\dfrac{1}{2}\left|\beta_1e^{i\theta_{1}}-\dfrac{1}{2} \right|^2\\
&&+
\dfrac{1}{2}\sum\limits_{(\vx,\vz)\neq(\vzero,\vzero),(\va,\vb)}
|\mu_{\vx,\vz}(M_{1})|^2,
\label{eq:m+p+}
\eeea
\beba
\Delta^2(M_{2},P_{2}(\va,\vb))
&\le & 
\dfrac{1}{2D}\|M_{2}-P_{2}(\va,\vb)\|_F^2\\
&=&
\dfrac{1}{2}\left|\alpha_2-\dfrac{1}{2}\right|^2
+
\dfrac{1}{2}\left|\beta_2e^{i\theta_{2}}+\dfrac{1}{2}\right|^2\\
&&+
\dfrac{1}{2}\sum\limits_{(\vx,\vz)\neq(\vzero,\vzero),(\va,\vb)}
|\mu_{\vx,\vz}(M_{2})|^2,
\label{eq:m-p-}
\eeea
\beba
\Delta^2(M_{i},P_{i}(\va,\vb)) 
&=& \dfrac{1}{2D} \|M_i\|_F^2\\
&=&
\dfrac{1}{2}\sum\limits_{\vx,\vz \in \Ztn} |\mu_{\vx,\vz}(M_{i})|^2,~~\forall i \ge 3.
\label{eq:mipi}
\eeea
Here $D=2^n$. Taking the sum of Eqs.(\ref{eq:m+p+}), (\ref{eq:m-p-})
and (\ref{eq:mipi}) for all $i \ge 3$ and using
Eq.(\ref{eq:tailmu}), we obtain
\beba
\Delta^2(M,P(\va,\vb))
&\le & 
\dfrac{1}{2}\left|\alpha_1-\dfrac{1}{2}\right|^2
+\dfrac{1}{2}\left|\beta_1e^{i\theta_{1}}-\dfrac{1}{2}\right|^2 \\
&&+\dfrac{1}{2}\left|\alpha_2-\dfrac{1}{2}\right|^2 
+\dfrac{1}{2}\left|\beta_2e^{i\theta_{2}}+\dfrac{1}{2}\right|^2\\
&&+2\gamma 
\label{eq:mpvavb}
\eeea
Now by Eq.(\ref{eq:alphabeta+-}) we have
\bea
\left|\alpha_1-\dfrac{1}{2}\right|^2 
\le \left|\alpha_1^2-\dfrac{1}{4}\right| 
\le \gamma,
\label{eq:a1} \\
\left|\alpha_2-\dfrac{1}{2}\right|^2 
\le \left|\alpha_2^2-\dfrac{1}{4}\right| 
\le \gamma.
\label{eq:a2}
\eea
Moreover, by Eqs.(\ref{eq:alphabeta+-})-(\ref{eq:theta-}), we have
\beba
\left|\beta_1e^{i\theta_{1}}-\dfrac{1}{2}\right|^2
&=&\beta_1^2+\dfrac{1}{4}-\beta_1\cos(\theta_1)\\
&\le & \left(\dfrac{1}{4}+\gamma\right)+\dfrac{1}{4} 
-\sqrt{\dfrac{1}{4}-\gamma}\cdot\\
&&(1-4\gamma-4\delta) \\
&\le & 5\gamma+2\delta, 
\label{eq:b1}
\eeea
\beba
\left|\beta_2e^{i\theta_{2}}+\dfrac{1}{2}\right|^2
&=&\beta_2^2+\dfrac{1}{4}+\beta_2\cos(\theta_2)\\
&\le & \left(\dfrac{1}{4}+\gamma \right)+\dfrac{1}{4} 
-\sqrt{\dfrac{1}{4}-\gamma}\cdot\\
&&(1-4\gamma-4\delta) \\
&\le & 5\gamma+2\delta.
\label{eq:b2}
\eeea
Plugging Eqs.(\ref{eq:a1})-(\ref{eq:b2}) into Eq.(\ref{eq:mpvavb}) yields
\bea
\Delta^2(M,P(\va,\vb))\le 8\gamma+2\delta
\eea
as desired.

\section{Proof of Lemma \ref{lem:klocal}}
\label{apd:lemklocal}

Suppose $|M|=k$, i.e. $M=\{M_1,M_2,\dots,M_k\}$. 
We first prove that
\bea
\sum\limits_{i=1}^k f_T^{\dagger}(M_i)
f_T(M_i) \le I.
\label{eq:ftft1}
\eea
(When we write $A \le B$ for two matrices $A,B$, we mean that $B-A$ is positive semidefinite.) Note that by the definition of $f_T(M_i)$, we can write it as
\bea
f_T(M_i)=\tilde{f}_T(M_i) \otimes I
\eea
where $\tilde{f}_T(M_i)$ is some operator on $T$, and $I$ is the identity operator on $T^c$. Then Eq.(\ref{eq:ftft1}) is equivalent to 
\bea
\sum\limits_{i=1}^k \tilde{f}_T^{\dagger}(M_i)
\tilde{f}_T(M_i) \le I.
\label{eq:ftft2}
\eea
Let $\ket{\psi}$ be an arbitrary pure state on $T$, and $\rho$ be the uniformly mixed state on $T^c$. Then by plugging
\bea
M_i=\tilde{f}_T(M_i)\otimes I+g(M_i)
\eea
into
\beba
1 &=& \tr (\ketbra{\psi}{\psi}\otimes \rho) \\
&=&\tr[(\sum\limits_{i=1}^k M_i^{\dagger}M_i) (\ketbra{\psi}{\psi}\otimes \rho )],
\eeea
and noting that $\tr (g_T(M_i))=0$ for any $i$, we obtain
\beba
1&=&\bra{\psi}(\sum\limits_{i=1}^k \tilde{f}_T^{\dagger}(M_i)
\tilde{f}_T(M_i) )\ket{\psi} \\
&&+\tr[ (\sum\limits_{i=1}^k g_T^{\dagger}(M_i)g_T(M_i)) (\ketbra{\psi}{\psi}\otimes \rho )] \\
& \ge &
\bra{\psi}(\sum\limits_{i=1}^k \tilde{f}_T^{\dagger}(M_i)
\tilde{f}_T(M_i)) \ket{\psi}.
\eeea
Since $\ket{\psi}$ is arbitrary, we get
\bea
\sum\limits_{i=1}^k \tilde{f}_T^{\dagger}(M_i)
\tilde{f}_T(M_i) \le I
\eea
as desired.

Now let $N=\{N_1,N_2,\dots,N_{k+1}\}$, where
\bea
N_i&=f_T(M_i),~~~\forall i=1,2,\dots,k;\\
N_{k+1}&=\sqrt{I-\sum\limits_{i=1}^k f^{\dagger}_T(M_i) f_T(M_i)}.stem 
\eea
(For a positive semidefinite matrix $A$ with spectral decomposition $A=\sum_x \lambda_x \ketbra{x}{x}$, we define 
$\sqrt{A} \coloneqq \sum_x \sqrt{\lambda_x} \ketbra{x}{x}$.)
Then $N$ is a valid measurement. Moreover, since each $f_T(M_i)$ acts non-trivially only on subsystem $T$, so does each $N_i$. Hence $N$ is a $|T|$-local measurement.

Now, the distance between $M_i$ and $N_i$ satisfies
\beba
\Delta^2(M_i,N_i) &\le &
\dfrac{1}{2D} \|M_i-N_i\|_F^2\\
&=&\dfrac{1}{2D} \|g_T(M_i)\|_F^2 \\
&=&\dfrac{1}{2D} (\|M_i\|_F^2-\|f_T(M_i)\|_F^2),
\eeea
for $i=1,2,\dots,k$, and
\beba
\Delta(M_{k+1},N_{k+1}) 
&=& 
\dfrac{1}{2D} \|N_{k+1}\|_F^2\\
&=&\dfrac{1}{2}-
\dfrac{1}{2D} 
\sum\limits_{i=1}^k \|f_T(M_i)\|_F^2.
\eeea
So the distance between $M$ and $N$ satisfies
\beba
\Delta^2(M,N)
&=& \sum\limits_{i=1}^{k+1} \Delta^2(M_i,N_i) \\
&\le & 1-\dfrac{1}{D}\sum\limits_{i=1}^k \|f_T(M_i)\|_F^2 \\
&\le & \delta^2.
\eeea
Here we have used the fact that 
\bea
\sum\limits_{i=1}^k \|M_i\|_F^2
=\tr \left(\sum\limits_{i=1}^k M_i^{\dagger}M_i \right)
=\tr (I)=D.
\eea

\section{Proof of Lemma \ref{lem:perminv}}
\label{apd:lemperminv}
The proof of this lemma is similar to that of lemma \ref{lem:klocal}.

Suppose $|M|=k$, i.e. $M=\{M_1,M_2,\dots,M_k\}$. We first show that 
\bea
\sum\limits_{i=1}^k \hat{M}_i^{\dagger}\hat{M}_i \le I,
\eea
or equivalently,
\bea
\sum\limits_{i=1}^k (\hat{M_i})_{\vla}^{\dagger} (\hat{M_i})_{\vla} \le I,~~\forall \vla \in \I_{d,n}.
\eea
Fix any $\vla \in \I_{d,n}$. Let $\ket{\psi}$ be an arbitrary pure state in $\W_{\vla,d}$, and $\rho$ be the uniformly mixed state in $\V_{\vla}$. Then by plugging
\bea
M_i=(\hat{M}_i+\tilde{M}_i+\bar{M}_i)_{\sch}
\eea
where
\bea
\hat{M_i} &= &\bigoplus\limits_{\vla \in \I_{d,n}} (\hat{M_i})_{\vla} \otimes I, \\
\tilde{M_i} &= & 
\bigoplus\limits_{\vla \in \I_{d,n}}
(\sum\limits_{j \neq 0} (\tilde{M_i})_{\vla,j} \otimes g_{\vla,j}),\\
\bar{M_i} &= &\sum\limits_{\vla \neq \vla'}
\ketbra{\vla}{\vla'} \otimes (M_i)_{\vla,\vla'}
\eea
into
\beba
1 &=& \tr[(\ketbra{\vla}{\vla} \otimes \ketbra{\psi}{\psi} \otimes \rho)_{\sch}] \\
&=& \tr\left[\left(\sum\limits_{i=1}^k M_i^{\dagger}M_i \right) (\ketbra{\vla}{\vla} \otimes \ketbra{\psi}{\psi} \otimes \rho)_{\sch}\right],
\eeea
and noting that $\tr (g_{\vla,j})=0$ for $j \neq 0$, we obtain
\beba
1&=& 
\bra{\psi}
(\sum\limits_{i=1}^k (\hat{M_i})_{\vla}^{\dagger}(\hat{M_i})_{\vla})
\ket{\psi} \\
&&+
\bra{\psi}
(\sum\limits_{i=1}^k \sum\limits_{j \neq 0}
(\tilde{M_i})_{\vla,j}^{\dagger}(\tilde{M_i})_{\vla,j})
\ket{\psi}\\
&&+ 
\tr[(\sum\limits_{i=1}^k \bar{M_i}^{\dagger}\bar{M_i})
(\ketbra{\vla}{\vla}\otimes
\ketbra{\psi}{\psi} \otimes \rho)]\\
&\ge &
\bra{\psi}
\sum\limits_{i=1}^k (\hat{M_i})_{\vla}^{\dagger}(\hat{M_i})_{\vla}
\ket{\psi}.
\eeea
Since $\ket{\psi}$ is arbitrary, we get
\bea
\sum\limits_{i=1}^k (\hat{M_i})_{\vla}^{\dagger} (\hat{M_i})_{\vla} \le I,
\eea
as desired.

Now let $N=\{N_1,N_2,\dots,N_{k+1}\}$ where
\bea
N_i&=&(\hat{M}_i)_{\sch},~~~\forall i=1,2,\dots,k;\\
N_{k+1}&=&\sqrt{I-\sum\limits_{i=1}^k (\hat{M}_i^{\dagger}\hat{M}_i)_{\sch}}.
\eea
Then $N$ is a valid measurement. In addition, since each $\hat{M_i}$ is permutation-invariant, so is each $N_i$. Thus, $N$ is a permutation-invariant measurement.

Now the distance between $M_i$ and $N_i$ satisfies
\beba
\Delta^2(M_i,N_i) &\le &
\dfrac{1}{2D}\|M_i-N_i\|_F^2 \\
&=&
\dfrac{1}{2D}\|(\tilde{M}_i+\bar{M}_i)_{\sch}\|_F^2 \\
&=&
\dfrac{1}{2D}(\|\tilde{M}_i\|_F^2+\|\bar{M}_i)\|_F^2)\\
&=&
\dfrac{1}{2D}(\|M_i\|_F^2-\|\hat{M}_i\|_F^2),
\eeea
for $i=1,2,\dots,k$, and
\beba
\Delta^2(M_{k+1},N_{k+1}) &\le &
\dfrac{1}{2D}\|N_{k+1}\|_F^2 \\
&=&
\dfrac{1}{2}-\dfrac{1}{2D}\sum\limits_{i=1}^k \|\hat{M}_i\|_F^2 \\.
\eeea
So the distance between $M$ and $N$ satisfies
\beba
\Delta^2(M,N)&=&\sum\limits_{i=1}^{k+1} \Delta^2(M_i,N_i) \\ 
& \le & 
1-\dfrac{1}{D}\sum\limits_{i=1}^k \|\hat{M}_i\|_F^2 \\
&\le & \delta^2.
\eeea
Here we have used the fact that 
\bea
\sum\limits_{i=1}^k \|M_i\|_F^2
=\tr \left(\sum\limits_{i=1}^k M_i^{\dagger}M_i \right)
=\tr (I)=D.
\eea

\section{Proof of Lemma \ref{lem:chimchin}}
\label{apd:lemchimchin}
For convenience, let $\I=\{i: L_i \ge 0.1 \delta^2L/k\}$ and $\xi_i=|\bracket{\tilde{v}(M_i)}{\tilde{v}(N_i)}|$ for any $i$. Then by $\Delta(M,N)\ge \delta$ we get
\beba
\delta^2 &\le& \Delta^2(M,N) \\
&=&1-\sum\limits_{i=1}^k |\bracket{v(M_i)}{v(N_i)}| \\
&=&1-\sum\limits_{i=1}^k \sqrt{p(M_i)p(N_i)}\xi_i.
\eeea
It follows that
\beba
1-\delta^2
& \ge &
\sum\limits_{i=1}^k \sqrt{p(M_i)p(N_i)}\xi_i \\
& \ge &
\sum\limits_{i \in \I} \sqrt{p(M_i)p(N_i)}\xi_i \\
& \ge &
(1-0.1\delta^2)\sum\limits_{i \in \I} \dfrac{L_i\xi_i}{L},
\eeea
since $p(M_i),p(N_i) \ge (1-0.1\delta^2)L_i/L$ for any $i \in \I$.
Hence
\beba
 \sum\limits_{i \in \I} {L_i\xi_i}
& \le &
\dfrac{(1-\delta^2)L}{1-0.1\delta^2}  
 \le 
(1-0.9\delta^2)L.
\label{eq:lxi1}
\eeea
Meanwhile, for any $i \not\in \I$, we have
$L_i<0.1\delta^2L/k$, and there are at most $k$ such $i$'s,
so 
\beba
\sum\limits_{i \in \I}L_i \ge (1-0.1\delta^2)L.
\label{eq:lxi2}
\eeea
Furthermore, by the fact that the arithmetic mean of a set of non-negative real numbers is no less than their geometric mean, we get
\beba
\prod\limits_{i \in \I} \xi_i^{L_i}
&\le &
\left(\dfrac{\sum_{i \in \I}L_i\xi_i}{\sum_{i \in \I}L_i}\right)^{\sum_{i \in \I}L_i}.
\label{eq:lxi3}
\eeea
Now, by Eqs.(\ref{eq:lxi1})-(\ref{eq:lxi3}) and $0 \le \xi_i \le 1$, we obtain
\beba
|\bracket{\chi_{\vL}(M)}{\chi_{\vL}(N)}|
&=&
\prod\limits_{i=1}^k \xi_i^{L_i} \\
&\le &
\prod\limits_{i \in \I} \xi_i^{L_i}\\
& \le & \left( \dfrac{1-0.9\delta^2}{1-0.1\delta^2} \right)^{(1-0.1\delta^2)L} \\
&\le & (1-0.8\delta^2)^{0.9L} \\
&\le & (1-0.6\delta^2)^L.
\eeea

\section{Proof of Lemma \ref{lem:pi}}
\label{apd:lempi}
Without loss of generality we assume that $\ket{\phi_1}$, $\ket{\phi_2}$,$\dots$, $\ket{\phi_m}$ are linearly independent. Then the subspace spanned by them is $m$-dimensional. Let $\{\ket{\varphi_1},
\ket{\varphi_2},\dots, \ket{\varphi_m}\}$ be an orthonormal basis for this subspace, where
\beba
\ket{\varphi_i}&=&\sum\limits_{j=1}^m
\lambda_{i,j} \ket{\phi_j}.
\eeea
for some $\lambda_{i,j}$'s. Then
\beba
m &=& \sum\limits_{i=1}^{m}\bracket{\varphi_i}{\varphi_i}\\
&=& \sum\limits_{i=1}^m
\sum\limits_{j,j'=1}^m
\lambda_{i,j}^* \lambda_{i,j'} \bracket{\phi_j}{\phi_{j'}} \\
&=&
\sum\limits_{i=1}^m
\sum\limits_{j=1}^m
|\lambda_{i,j}|^2
+
\sum\limits_{i=1}^m
\sum\limits_{j \neq j'}
\lambda_{i,j}^* \lambda_{i,j'} \bracket{\phi_j}{\phi_{j'}}\\
&\ge &
\sum\limits_{i=1}^m
\sum\limits_{j=1}^m
|\lambda_{i,j}|^2-
\sum\limits_{i=1}^m
\sum\limits_{j \neq j'}
\dfrac{|\lambda_{i,j}|^2+|\lambda_{i,j'}|^2}{2}\cdot\dfrac{1}{5m} \\
&\ge &
\dfrac{1}{2}\sum\limits_{i=1}^m
\sum\limits_{j=1}^m
|\lambda_{i,j}|^2.
\eeea
Hence
\beba
\bra{\psi}\Pi\ket{\psi} &=& \sum\limits_{i=1}^m
|\bracket{\varphi_i}{\psi}|^2\\
&=&
\sum\limits_{i=1}^m
|\sum\limits_{j=1}^m \lambda_{i,j}^* \bracket{\phi_j}{\psi}|^2 \\
&=&
\sum\limits_{i=1}^m
\sum\limits_{j,j'=1}^m \lambda_{i,j}^*\lambda_{i,j'} \bracket{\phi_j}{\psi}
\bracket{\psi}{\phi_{j'}} \\
&\le &
\sum\limits_{i=1}^m
\sum\limits_{j,j'=1}^m
|\lambda_{i,j}||\lambda_{i,j'}|
\cdot \left(\dfrac{1}{5m} \right )^2
\\
&\le &
\dfrac{1}{25m^2}
\sum\limits_{i=1}^m
\sum\limits_{j,j'=1}^m
\dfrac{|\lambda_{i,j}|^2+|\lambda_{i,j'}|^2}{2} \\
&\le &
\dfrac{1}{25m}\sum\limits_{i=1}^m
\sum\limits_{j=1}^m
|\lambda_{i,j}|^2 \\
&\le & 0.1.
\eeea

\end{appendix}

\end{document}